\documentclass[a4paper]{article}
\advance\textwidth by 2cm
\advance\oddsidemargin by -1cm

\usepackage{graphicx}
\usepackage{color}
\usepackage{bcprules}
\usepackage{amsmath}
\usepackage{amssymb}
\usepackage{url}
\usepackage{multicol}

\newif\ifdraft\draftfalse
\newif\iffull\fullfalse
\ifdraft
\newcommand\nk[1]{\textcolor{red}{[#1 -nk]}}
\newcommand\ry[1]{\textcolor{blue}{[#1 -ry]}}
\newcommand\fix[1]{{\color{red}{#1}}}
\else
\newcommand\nk[1]{}
\newcommand\ry[1]{}
\newcommand\fix[1]{#1}
\fi

\newcommand\Bot{\mathop{\bot\!\!\!\bot}}
\newcommand\Top{\mathop{\top\!\!\!\top}}
\newcommand\Forall{\texttt{Forall}}
\newcommand\Exists{\texttt{Exists}}
\newcommand\figref[1]{Fig.~\ref{#1}}
\newcommand\tabref[1]{Table~\ref{#1}}
\newcommand\DWF{\mathtt{DWF}}
\newcommand\LLRF{\mathtt{WF}}
\newcommand\MC{\mathtt{MC}}
\newcommand\Nat{\mathbf{N}}
\newcommand\Base[1]{\mathbf{B}_{#1}}
\newcommand\measure[1]{\#{#1}}
\newcommand\meas[2]{\#_{#2}{#1}}
\newcommand\Down[1]{#1^{\downarrow}}
\newcommand\IF{\texttt{if}}

\newcommand\THEN{\texttt{then}}
\newcommand\ELSE{\texttt{else}}
\newcommand\MOD{\texttt{mod}}

\newcommand\muhfl{\textsc{MuHFL}}
\newcommand\sem[1]{\mathbin{[\![}#1\mathbin{]\!]}}
\newcommand\semd[1]{\mathbin{(\!|}#1\mathbin{|\!)}}
\newcommand\Simge{\succeq}

\newcommand\LEmono{\leq}
\newcommand\LEQ{\sqsubseteq}
\newcommand\GEQ{\sqsupseteq}
\newcommand\LUB{\sqcup}
\newcommand\GLB{\sqcap}
\newcommand\GFP{\mathbf{GFP}}
\newcommand\LFP{\mathbf{LFP}}
\newcommand\Z{\mathbf{Z}}
\newcommand\hflz{\textsc{HFL(Z)}}
\newcommand\p{\vdash}
\newcommand\pST{\p_{\mathtt{ST}}}
\newcommand\pSTex{\pST}
\newcommand\stenv{\Gamma}
\newcommand\ST{\mathtt{ST}}
\newcommand\nuhflz{\textsc{$\nu$HFL(Z)}}
\newcommand\FALSE{\mathtt{false}}
\newcommand\TRUE{\mathtt{true}}
\newcommand\sty{\kappa}
\newcommand\fty{\tau}
\newcommand\ord{\mathtt{ord}}
\newcommand\ar{\mathtt{ar}}
\newcommand\INT{\mathtt{Int}}
\newcommand\Prop{\star}
\newcommand\form{\varphi}
\newcommand\ctlstar{CTL$^*$}
\newcommand\fib{\mathit{Fib}}
\newcommand\fibsafe{\mathit{Fib}_{\mathtt{safe}}}
\newcommand\imp{\Rightarrow}
\newcommand\IFF{\Leftrightarrow}
\newcommand\pahfl{{\sc PaHFL}}

\newcommand\rethfl{{\sc ReTHFL}}
\newcommand\defeq{\stackrel{\triangle}{=}}
\newcommand\seq[1]{\widetilde{#1}}
\newcommand\set[1]{\{#1\}}
\newcommand\COL{\mathbin{:}}
\newcommand\exv[1]{v_{#1}} 
\newcommand\cv{u} 
\newcommand\exarg{\mathit{exarg}}
\newcommand\letexp[2]{\mathbf{let}\ #1=#2\ \mathbf{in}\ }
\newcommand\raweq{\approx}

\newcommand\efty{\zeta}
\newcommand\tagty{\alpha}
\newcommand\trT[1]{#1^\dagger}
\newcommand\tagv{t}

\newcommand\FV{\mathbf{FV}}
\newcommand\dom{\mathit{dom}}
\newcommand\restrict[2]{#1\mathbin{\downarrow_{#2}}}
\newcommand\etyfun[3]{(#1,#2)\to #3}
\newcommand\intfun[1]{\INT\to #1}
\newcommand\taguse{\mathrm{T}}
\newcommand\tagnotuse{\mathrm{F}}

\newcommand\judgesimp[4]{\ensuremath{{#1}\vdash{#3}\colon{#4}\ \leadsto\ }} 
\newcommand\judgesimpc[5]{\ensuremath{{#1}\vdash^{#5}{#3}\colon{#4}\ \leadsto\ }} 

\newcommand\envv{\Delta}

\newcommand\envpair[3]{#1:(#2,#3)}

\newcommand\abs[3]{\lambda{#1}^{#2}.\,#3}

\newcommand\gettags{\mathrm{Tags}}

\newcommand\subtype{<:}

\newtheorem{theorem}{Theorem}[section]
\newtheorem{lemma}[theorem]{Lemma}
\newtheorem{remark}{Remark}
\newtheorem{notation}{Notation}
\newtheorem{example}{Example}
\newcommand\qed{$\Box$}
\newenvironment{proof}{\paragraph{Proof}}{\hfill\qed}

\begin{document}

\title{Automatic HFL(Z) Validity Checking for Program Verification}

\author{Naoki Kobayashi \and Kento Tanahashi \\
  The University of Tokyo
 \and Ryosuke Sato\and Takeshi Tsukada \\ Chiba University}




\maketitle
\begin{abstract}
  We propose an automated method for checking the validity of
  a formula of \hflz{}, a higher-order logic with fixpoint operators and integers.
  Combined with Kobayashi et al.'s reduction from higher-order program verification
   to \hflz{} validity checking, our method yields a fully automated, uniform
   verification method for arbitrary temporal properties of higher-order functional programs
   expressible in the modal \(\mu\)-calculus, 
   including termination, non-termination, fair termination, fair non-termination,
   and also branching-time properties.
  We have implemented our method and obtained promising experimental results.
\end{abstract}

\section{Introduction}
\label{sec:intro}
Kobayashi et al.~\cite{KTW18ESOP,DBLP:conf/pepm/WatanabeTO019}
have shown that temporal property verification problems for higher-order functional
programs can be reduced to
the validity checking problem for \hflz{}. \hflz{} is an extension of Viswanathan and
Viswanathan's higher-order fixpoint logic (HFL)~\cite{Viswanathan04} with integers, and the validity checking
problem asks whether or not a given \hflz{} formula (without modal operators) is valid.
The reduction provides a uniform approach to the temporal property verification
of higher-order functional programs. Automatic validity checkers have been implemented for
the first-order fragment of \hflz~\cite{DBLP:conf/sas/0001NIU19}
and \nuhflz{}~\cite{DBLP:conf/sas/IwayamaKST20,DBLP:conf/aplas/KatsuraIKT20},
the fragment of \hflz{} without least fixpoint operators.
The former~\cite{DBLP:conf/sas/0001NIU19}
enables automated verification of temporal properties of first-order programs,
and the latter~\cite{DBLP:conf/sas/0001NIU19,DBLP:conf/sas/IwayamaKST20,DBLP:conf/aplas/KatsuraIKT20}
enables automated verification of safety properties of higher-order programs.
This line of work provides a streamlined, general approach to automated verification of temporal
properties of programs. Despite the generality of the approach, it has been reported that
Kobayashi et al.'s tool~\cite{DBLP:conf/sas/0001NIU19}
outperformed Cook and Koskinen's method specialized for CTL verification~\cite{Cook2013a}.

 Following the above line of research, we propose an automated (sound but
 incomplete\footnote{Incompleteness is inevitable because the validity checking problem
   for \hflz{} is undecidable in general.})
   method of validity
checking for the \emph{full} fragment of \hflz{}. By combining the proposed method with
the above-mentioned reduction~\cite{KTW18ESOP,DBLP:conf/pepm/WatanabeTO019}, we can obtain 
a fully automated verification method for arbitrary regular temporal properties of higher-order programs.
The properties that can be verified in a uniform manner using our method
include safety~\cite{Jhala08,Terauchi10POPL,KSU11PLDI,SUK13PEPM,Ong11POPL,DBLP:journals/pacmpl/PavlinovicSW21,zhu_2015},
termination~\cite{Kuwahara2014Termination}, non-termination~\cite{Kuwahara2015Nonterm},
fair termination~\cite{MTSUK16POPL}, and fair non-termination~\cite{Watanabe16ICFP},
for which separate methods and tools have been developed so far.
Furthermore, our method can also be used for automatic verification
of branching-time properties (typically expressed by formulas of CTL, CTL*, and
the modal \(\mu\)-calculus),
which have not been supported by previous automated methods/tools for
higher-order program verification, to our knowledge.

\tabref{tab:fixpoint-logics} compares \hflz{} with
other fixpoint logics studied in the context of automated program verification.
Program verification by reduction to the satisfiability problem of
 constrained Horn clauses (CHCs) has recently been studied actively as a uniform
method for automated verification of first-order programs~\cite{Bjorner15}.
As discussed in \cite{DBLP:conf/sas/0001NIU19},
the CHC satisfiability problem corresponds to the validity checking problem
for the first-order fragment of \hflz{} with only the greatest fixpoint operators;
thus CHCs can be used to verify safety properties of first-order programs, but
extensions~\cite{DBLP:conf/cav/BeyenePR13,DBLP:conf/sas/BjornerMR13}
are required to reason about other properties such as
liveness and termination. The Mu-Arithmetic studied by Kobayashi
et al.~\cite{DBLP:conf/sas/0001NIU19} allows arbitrary alternations of
greatest and least fixpoint operators, and can be used for verification of
arbitrary regular properties of first-order programs, but not higher-order ones.
Burn et al.~\cite{DBLP:journals/pacmpl/BurnOR18} studied a higher-order extension
of CHCs, and Katsura
et al.~\cite{DBLP:conf/sas/IwayamaKST20,DBLP:conf/aplas/KatsuraIKT20}
studied the corresponding fragment of \hflz{} called \nuhflz{}. They
can be used for verification of safety properties of
higher-order programs, but not arbitrary temporal properties.
The automated method for \hflz{} validity checking developed in this paper
enables a uniform approach to automated verification of arbitrary regular
temporal properties of higher-order programs.

\begin{table}[tbp]
  \caption{Fixpoint logics for program verification}
  \label{tab:fixpoint-logics}
  \begin{tabular}{|l|l|l|}
    \hline
    & \(\mu\)- or \(\nu\)-only & both \(\mu\) and \(\nu\)\\
    \hline
    first-order  &
    \begin{minipage}{0.38\textwidth}
      Constrained Horn Clauses (CHCs)\\
      \cite{Bjorner15,DBLP:conf/vmcai/JaffarSV06,DBLP:journals/sttt/DelzannoP01}
      \end{minipage} & Mu-Arithmetic~\cite{DBLP:conf/sas/0001NIU19}\\
\hline    
higher-order &
    \begin{minipage}{0.38\textwidth}
Higher-order CHCs~\cite{DBLP:journals/pacmpl/BurnOR18},
\nuhflz{}~\cite{DBLP:conf/sas/IwayamaKST20,DBLP:conf/aplas/KatsuraIKT20}
\end{minipage}
    &
    \begin{minipage}{0.38\textwidth}
    \hflz{}~\cite{KTW18ESOP,DBLP:conf/pepm/WatanabeTO019}
\end{minipage}\\
\hline
\end{tabular}
\end{table}

Our approach to automated validity checking of \hflz{} formulas has been
inspired by the approach of Kobayashi et al.~\cite{DBLP:conf/sas/0001NIU19} for
first-order \hflz{}
and that of Fedyukovich et al.~\cite{freqterm} for termination analysis
of first-order programs.
We approximate a given \hflz{} formula with a formula of \nuhflz{},
the fragment of \hflz{} without the least-fixpoint operator. We can then use
existing solvers~\cite{DBLP:conf/sas/IwayamaKST20,DBLP:conf/aplas/KatsuraIKT20} to prove
the validity of the \nuhflz{} formula.
\nk{The following explanation may be too technical for
  introduction. Should we give an overview of the method after
  introducing the syntax in the next section?}
The idea of removing the least-fixpoint operator is as follows.
Suppose we wish to prove that \((\mu X.\varphi(X))\,y\) holds for every integer \(y\),
where \(\mu X.\varphi(X)\) represents the \emph{least} predicate
such that \(X = \varphi(X)\);
for example, \(\mu X.\lambda z.X(z)\) is equivalent to \(\lambda z.\FALSE\), and
\(\mu X.\lambda z.z=0\lor X(z-1)\) is equivalent to \(\lambda z.z\ge 0\).
By a standard property of least fixpoints, \(\mu X.\varphi(X)\) can be underapproximated by
a formula of the form
\(\varphi^e(\lambda z.\FALSE) \equiv \underbrace{\varphi(\cdots \varphi}_e(\lambda z.\FALSE)\cdots)\),
where \(e\) is an expression denoting a non-negative integer.
(The formula \(\varphi^e(\lambda z.\FALSE)\) is actually represented by using
  the greatest fixpoint operator, as discussed later.)
We then use an existing \nuhflz{} validity checker to check the validity of \((\varphi^e(\lambda z.\FALSE))\,y\).
If \((\varphi^e(\lambda z.\FALSE))\,y\) is valid, then we can conclude that the original
formula \((\mu X.\varphi(X))\,y\) is also valid.
Otherwise, we increase the value of \(e\) to improve the precision
and run the \nuhflz{} validity checker again.
Following Kobayashi et al.'s work on the first-order case~\cite{DBLP:conf/sas/0001NIU19},
we consider as \(e\) an expression of the form \(c_0+c_1|x_1|+\cdots+c_k|x_k|\), where \(x_1,\ldots,x_k\)
are the integer variables in scope, and gradually increase the coefficients \(c_0,\ldots,c_k\).

The new challenge in this paper for dealing with the higher-order case is how to incorporate
the values of higher-order variables into \(e\).
To this end, for each function argument, we add an extra integer argument that represents information
about the function argument (thus, a predicate of the form \(\lambda f.\varphi\) would be transformed to
\(\lambda (\exv{f},f).\varphi'\), where \(\exv{f}\) is the extra integer argument that represents information
about \(f\), and used in \(\varphi'\) to compute the value of \(e\) above).
The idea of adding extra arguments has been inspired by the work of Unno et al.~\cite{UnnoTK13} on
relatively complete verification of safety properties of higher-order functional programs,
but we have devised a different, more systematic method for inserting extra arguments.
To avoid the insertion of unnecessary extra arguments, we also propose a type-based static analysis
to estimate necessary extra arguments.

The contributions of this paper are summarized as follows.
\begin{enumerate}
\item An extension of Kobayashi et al.'s method for the first-order \hflz{}~\cite{DBLP:conf/sas/0001NIU19},
  to obtain an automated validity checking method for full \hflz{}.
\item A method of adding extra arguments for higher-order arguments, to improve the precision.
\item An optimization to avoid the insertion of unnecessary extra arguments.
\item A theoretical characterization of the power of our method (Section~\ref{sec:disc}).
  We compare our method with previous popular methods for proving termination,
  such as those using lexicographic linear ranking functions and disjunctive well-founded relations.
\item An implementation and experiments on the proposed methods above (Section~\ref{sec:exp}). According to the experiments,
  our tool outperformed previous verification tools specialized 
  for 
  verification of termination~\cite{Kuwahara2014Termination}, non-termination~\cite{Kuwahara2015Nonterm},
  fair termination~\cite{MTSUK16POPL}, and fair non-termination~\cite{Watanabe16ICFP}.
  We have also confirmed that our tool can verify properties of higher-order programs
  that were not supported by previous automated tools,
  including branching-time properties. 
\end{enumerate}

The rest of this paper is structured as follows.
Section~\ref{sec:pre} reviews \hflz{} and its connection to program verification.
Section~\ref{sec:method} describes our method for \hflz{} validity checking.
Section~\ref{sec:disc} gives some theoretical characterization of the power of our method
by comparing it with previous methods for proving termination and liveness properties.
Section~\ref{sec:exp} reports experimental results.
Section~\ref{sec:related} discusses related work and Section~\ref{sec:conc} concludes the paper.

\section{Preliminaries}
\label{sec:pre}

This section reviews \hflz{} and its application to program verification.
\hflz{} is an extension of Viswanathan and
Viswanathan's higher-order fixpoint logic (HFL)~\cite{Viswanathan04}\footnote{We omit modal operators in this paper. The modal operators are unnecessary for the general reduction from program verification problems~\cite{DBLP:conf/pepm/WatanabeTO019}.}
with integers.

\subsection{\hflz{}}
The set of types, ranged over by \(\sty\), is given by:
\[
\begin{array}{l}
\sty \mbox{ (types)} ::= \INT \mid \fty \qquad 
\fty \mbox{ (predicate types)} ::= \Prop \mid \sty \to \fty.
\end{array}
\]
Here, \(\Prop\) is the type of propositions, and \(\INT\) is the type of integers.
A predicate type \(\fty\) is of the form \(\sty_1\to\cdots \to \sty_k\to \Prop\),
which describes \(k\)-ary (possibly higher-order) predicates on values of types \(\sty_1,\ldots,\sty_k\).
For example, \(\INT\to (\INT\to\Prop)\to \Prop\) is the type of binary predicates that takes
an integer and a predicate on integers as arguments.
For a type \(\sty\), we define the order and arity of \(\sty\), written
\(\ord(\sty)\) and \(\ar(\sty)\) respectively, by:
\[
\begin{array}{l}
\ord(\INT)=\ord(\Prop)=0\qquad
\ord(\sty\to\fty) = \max(\ord(\fty), \ord(\sty)+1)\\
\ar(\INT)=\ar(\Prop)=0\qquad
\ar(\sty\to\fty) = \ar(\fty)+1.
\end{array}
\]

The syntax of \hflz{} formulas is given as follows.

\[
  \begin{array}{l}
    \form \mbox{ (formulas) } ::=
    x \mid \form_1\lor \form_2 \mid \form_1\land\form_2\\
\qquad\qquad  \qquad     \mid \mu x^\fty.\form \mid \nu x^\fty.\form \quad\mbox{ (fixpoint operators) }\\
\qquad\qquad  \qquad     \mid \form_1\form_2 \mid \lambda x^\sty.\form\quad\ \ \mbox{ (\(\lambda\)-abstractions and applications)}\\
\qquad\qquad\qquad  \mid \form\, e \mid e_1\fix{\ge} e_2\qquad
  \mbox{ (extension with integers)}\\
  \ \   e \mbox{ (integer expressions) }::= n \mid x \mid  e_1+e_2 \mid e_1\times e_2
\end{array}
\]
Here, \(x\) and \(n\) are metavariables for
variables and integers respectively.
The formulas \(\mu x^\fty.\form\)  and \(\nu x^\fty.\form\) respectively
denote the least and greatest predicates \(x\) such that \(x=\form\).
For example, \(\mu x^{\Prop}.x\) and \(\nu x^{\Prop}.x\) are equivalent to \(\FALSE\)
(which can be expressed as \fix{\(0\ge 1\)})
and \(\TRUE\) (which can be expressed as \fix{\(0\ge 0\)}) respectively.
The variable \(x\) is bound in \(\mu x^\fty.\form\), \(\nu x^\fty.\form\), and \(\lambda x^\fty.\form\).
As usual, we implicitly assume \(\alpha\)-renaming of bound variables. We write
\([\form_1/x]\form_2\) for the capture-avoiding substitution of \(\form_1\) for
all the free occurrences of \(x\) in \(\form_2\).
We often omit the type annotation.
Henceforth, we often use shorthand notations like
\(e_1=e_2\) (for \fix{\(e_1\ge e_2\land e_2\ge e_1\)}) and \(e_1-e_2\) (for \(e_1 + (-1)\times e_2\))
and treat them as if they were primitives.

We consider only formulas well-typed under
the simple type system given in \figref{fig:st}.
In the figure, \(\stenv\) denotes a type environment of the form \(x_1\COL\sty_1,\ldots,x_k\COL\sty_k\),
which is considered a function that maps \(x_i\) to \(\sty_i\) for \(i\in\set{1,\ldots,k}\).
For example, \(\mu x^{\INT\to\Prop}.x\,1\) is rejected as ill-typed (since
the fixpoint variable \(x\) and \(x\,1\) have different types, violating \rn{T-Mu}).
\begin{figure}[tbp]
  \begin{multicols}{2}
    \typicallabel{T-Plus}
  \infrule[T-Var]{}{\stenv,x\COL\sty \pST x:\sty}
  \infrule[T-Or]{\stenv\pST \form_1:\Prop\andalso \stenv\pST \form_2:\Prop}
  {\stenv\pST\form_1\lor\form_2:\Prop}
  \infrule[T-And]{\stenv\pST \form_1:\Prop\andalso \stenv\pST \form_2:\Prop}
  {\stenv\pST\form_1\land\form_2:\Prop}
  \infrule[T-Mu]{\stenv,x\COL\fty\pST \form:\fty}
          {\stenv\pST\mu x^\fty.\form:\fty}
  \infrule[T-Nu]{\stenv,x\COL\fty\pST \form:\fty}
          {\stenv\pST\nu x^\fty.\form:\fty}
  \infrule[T-App]{\stenv\pST \form_1:\fty_2\to\fty\andalso \stenv\pST\form_2:\fty_2}
          {\stenv\pST\form_1\form_2:\fty}
          
  \infrule[T-Abs]{\stenv,x\COL\sty\pST\form:\fty}
          {\stenv\pST\lambda x^\sty.\form:\sty\to\fty}
  \infrule[T-AppInt]{\stenv\pST \form:\INT\to\fty\andalso \stenv\pST e:\INT}
          {\stenv\pST\form\,e:\fty}
  \infrule[T-Ge]{\stenv\pST e_1:\INT\andalso \stenv\pST e_2:\INT}
          {\stenv\pST e_1\fix{\ge} e_2:\Prop}
  \infrule[T-Int]{}
          {\stenv\pST n:\INT}
  \infrule[T-Plus]{\stenv\pST e_1:\INT\andalso \stenv\pST e_2:\INT}
          {\stenv\pST e_1+ e_2:\INT}
  \infrule[T-Mult]{\stenv\pST e_1:\INT\andalso \stenv\pST e_2:\INT}
          {\stenv\pST e_1\times e_2:\INT}
\end{multicols}
  \caption{Simple Type System for \hflz{}}
  \label{fig:st}
\end{figure}

\begin{notation}
  \label{not:eq}
For readability, we sometimes represent fixpoint formulas by using fixpoint equations.
For example, we call \(\alpha x.\lambda y.\form\) (where \(\alpha\) is \(\mu\) or \(\nu\))
``the predicate \(x\) defined by the equation
\(x\,y=_\alpha \form\)''. The latter presentation of fixpoint formulas is in general
called \emph{hierarchical equation systems} (HES)~\cite{DBLP:conf/popl/KobayashiLB17,KTW18ESOP}.
We omit the formal definition of the HES representation and use it only informally in this paper.
\qed
\end{notation}

\fix{We review the formal semantics of \hflz{} formulas.
For each simple type \(\sty\), we define the partially ordered set
\(\sem{\sty}=(\semd{\sty}, \LEQ_{\sty})\) by:
\[
\begin{array}{l}
  \semd{\INT}=\Z \qquad \LEQ_{\INT}=\set{(n,n)\mid n\in \Z})\\
  \semd{\Prop}=\set{\Bot,\Top}\qquad \LEQ_{\Prop}=\set{(\Bot,\Bot),(\Bot,\Top),(\Top,\Top)}\\
  \semd{\sty\to\fty} =
  \set{f\in \semd{\sty}\to\semd{\fty}\mid
    \forall x,y\in \semd{\sty}.x\LEQ_{\sty}y \imp f(x)\LEQ_{\fty} f(y)}\\
  \LEQ_{\sty\to\fty} =
  \set{(f,g)\in \semd{\sty\to\fty}\times\semd{\sty\to\fty}\mid
    \forall x\in \semd{\sty}.f(x)\LEQ_{\fty} g(x)}.\\
\end{array}
\]
Here, \(\Z\) denotes the set of integers.
For each \(\fty\), \(\sem{\fty}\) (but not \(\sem{\INT}\)) forms a complete lattice.
We write \(\Bot_\fty\) (\(\Top_\fty\)) for
the least (greatest, resp.) element of \(\sem{\fty}\), and
\(\GLB_\fty\) (\(\LUB_\fty\), resp.) for the greatest lower bound (least upper bound, resp.) operation with respect to \(\LEQ_\fty\).
We also define the least and greatest fixpoint operators
\(\LFP_\fty,\GFP_\fty \in \semd{(\fty\to\fty)\to\fty)}\) by:
\[
\begin{array}{l}
  \LFP_\fty(f) = \GLB \set{g\in\semd{\fty}\mid f(g)\LEQ_\fty g}\qquad
  \GFP_\fty(f) = \LUB \set{g\in\semd{\fty}\mid g\LEQ_\fty f(g)}.
\end{array}
\]
Note that \(\LFP_\fty\) and \(\GFP_\fty\) are well-defined,
since every element of \(\semd{\fty\to\fty}\) is a monotonic function over a complete lattice.
By Tarski's fixpoint theorem, \(\LFP_\fty(f)\) and \(\GFP_\fty(f)\) coincide with the least and greatest fixpoint of \(f\), respectively.

For a simple type environment \(\stenv\), we write \(\semd{\stenv}\)
for the set of maps \(\rho\) such that
\(\dom(\rho)=\dom(\stenv)\) and \(\rho(x)\in\semd{\stenv(x)}\) for each
\(x\in\dom(\rho)\).

For each valid type judgment \(\stenv\pST \form:\sty\), 
its semantics \(\sem{\stenv\pST \form:\sty}\in \semd{\stenv}\to \semd{\sty}\)
is defined by:
\[
\begin{array}{l}
  \sem{\Gamma,x\COL\sty\pST x\COL\sty}(\rho) = \rho(x)\\
  \sem{\Gamma\pST\form_1\lor\form_2:\Prop}{\rho}
  = \sem{\Gamma\pST\form_1:\Prop}{\rho}\LUB_{\Prop}
  \sem{\Gamma\pST\form_2:\Prop}{\rho}\\
  \sem{\Gamma\pST\form_1\land\form_2:\Prop}{\rho}
  = \sem{\Gamma\pST\form_1:\Prop}{\rho}\GLB_{\Prop}
  \sem{\Gamma\pST\form_2:\Prop}{\rho}\\
  \sem{\Gamma\pST \mu x^\fty.\form:\fty}{\rho}
  = \LFP(\lambda v\in\semd{\fty}.\sem{\Gamma,x\COL\fty\pST\form:\fty}(\rho\set{x\mapsto v}))\\
  \sem{\Gamma\pST \nu x^\fty.\form:\fty}{\rho}
  = \GFP(\lambda v\in\semd{\fty}.\sem{\Gamma,x\COL\fty\pST\form:\fty}(\rho\set{x\mapsto v}))\\
  \sem{\Gamma\pST \lambda x^\sty.\form:\sty\to\fty}{\rho}
  = \lambda w\in\semd{\sty}.\sem{\Gamma,x\COL\sty\pST\form:\fty}(\rho\set{x\mapsto w})\\
  \sem{\Gamma\pST \form_1\form_2:\fty}{\rho}
  =
  \sem{\Gamma\pST \form_1:\fty_2\to \fty}{\rho}\,
  (\sem{\Gamma\pST \form_2:\fty_2}{\rho})\\
  \sem{\Gamma\pST \form\,e:\fty}{\rho}
  =
  \sem{\Gamma\pST \form:\INT\to \fty}{\rho}\,
  (\sem{\Gamma\pST e:\INT}{\rho})\\
  \sem{\Gamma\pST e_1\ge e_2:\Prop}{\rho}
  = \left\{\begin{array}{ll}
  \Top &\mbox{if $\sem{\Gamma\pST e_1:\INT}{\rho}\ge\sem{\Gamma\pST e_2:\INT}{\rho}$}\\
  \Bot &\mbox{otherwise}
  \end{array}\right.\\
  \sem{\Gamma\pST n:\INT}{\rho}=n\\
  \sem{\Gamma\pST e_1+ e_2:\INT}{\rho}
  =\sem{\Gamma\pST e_1:\INT}{\rho}+\sem{\Gamma\pST e_2:\INT}{\rho}\\
  \sem{\Gamma\pST e_1\times e_2:\INT}{\rho}
  =\sem{\Gamma\pST e_1:\INT}{\rho}\times\sem{\Gamma\pST e_2:\INT}{\rho}\\
\end{array}
\]}
For a closed formula \(\form\) of type \(\fty\), we often just write \(\sem{\form}\) for
\(\sem{\emptyset\pST\form:\fty}\).
We write \(\form_1\equiv_{\Gamma,\fty} \form_2\) when
\(\sem{\Gamma\pST \form_1:\fty}=\sem{\Gamma\pST \form_2:\fty}\),
and often omit the subscripts \(\Gamma\) and \(\fty\).
Note in particular that the following laws hold (under an appropriate assumption on types):
(i) \(\alpha x.\form \equiv [\alpha x.\form/x]\form\) for \(\alpha\in\set{\mu,\nu}\),
(ii) \((\lambda x.\form_1)\form_2 \equiv [\form_2/x]\form_1\) (\(\beta\)-equality),
and (iii) \(\form \equiv \lambda x.\form\,x\) (\(\eta\)-equality).

\begin{example}
  \label{ex:simple}
  Consider the formula \(\mu x^{\INT\to\Prop}.\lambda y.y=0\lor x(y-1)\),
  which denotes the least predicate on integers that satisfies the equivalence
  \(x\equiv \lambda y.y=0\lor x(y-1)\).
  To understand what the formula means, let us expand the equality as follows.
  \[
  \begin{array}{lll}
    x&\equiv \lambda y.y=0 \lor x(y-1)\\
    & \equiv \lambda y.y=0 \lor (\lambda y.y=0 \lor x(y-1))(y-1) &\mbox{(expand $x$)}\\
    & \equiv \lambda y.y=0 \lor y-1=0 \lor x(y-2) & \mbox{($\beta$-reduction)}\\
    & \equiv \lambda y.y=0 \lor y-1=0 \lor (\lambda y.y=0 \lor x(y-1))(y-2)\quad &\mbox{(expand $x$)}\\
    & \equiv \lambda y.y=0 \lor y-1=0 \lor y-2=0 \lor x(y-3)& \mbox{($\beta$-reduction)}\\
    & \equiv \cdots.
  \end{array}
  \]
  Thus, the formula \(\mu x^{\INT\to\Prop}.\lambda y.y=0\lor x(y-1)\) is equivalent to \(\lambda y.y\ge 0\). \qed
\end{example}

\begin{example}
  \label{ex:quantifiers}
  Consider the formula
  \(\nu x^{(\INT\to\Prop)\to\Prop}.\lambda p.p(0)\land x(\lambda y.p(y+1))\),
  which represents the greatest predicate \(x\)
  such that
  \(x(p) \equiv p(0)\land x(\lambda y.p(y+1))\) for every unary predicate
  \(p\) on integers. We can expand the equality as follows.
    \[
  \begin{array}{ll}
    x(p) &\equiv p(0) \land (\lambda p.p(0)\land x(\lambda y.p(y+1)))(\lambda y.p(y+1))\\
    &\equiv p(0) \land p(1) \land x(\lambda y.p(y+2))
    \equiv p(0) \land p(1) \land p(2)\land x(\lambda y.p(y+3))
    \equiv \cdots.
  \end{array}
  \]
  Thus, 
  \(\nu x^{(\INT\to\Prop)\to\Prop}.\lambda p.p(0)\land x(\lambda y.p(y+1))\)
  is equivalent to \(\lambda p.\forall y\ge 0.p(y)\).
  In this manner, universal quantifiers can be expressed by using
  the greatest fixpoint operator \(\nu\) (note that
  \(\forall y.p(y)\) can be expressed as \(\forall y\ge 0.p(y)\land p(-y)\)).
  Similarly, existential quantifiers can be expressed by using
  the least fixpoint operator \(\mu\). Henceforth, we use
  quantifiers as if they were primitives.
  \qed
\end{example}

The \emph{validity checking problem} for \hflz{} (or, the \hflz{} validity checking problem)
asks 
whether a given (closed) \hflz{} formula \(\form\) is valid (i.e. whether \(\sem{\form}=\Top\)).
The problem is undecidable in general; we aim to develop an incomplete but sound
method for proving or disproving the validity of \hflz{}.

The fragment of \hflz{} without the least fixpoint operator \fix{\(\mu\)}
is called \nuhflz{}. In this paper, we shall develop an automated method
for \hflz{} validity checking, by reduction to \nuhflz{} validity checking;
for \nuhflz{} validity checking, a few tools
are available~\cite{DBLP:journals/pacmpl/BurnOR18,DBLP:conf/sas/IwayamaKST20,DBLP:conf/aplas/KatsuraIKT20}.


\subsection{Applications of \hflz{} to Program Verification}
\label{sec:app}
Watanabe et al.~\cite{DBLP:conf/pepm/WatanabeTO019} have shown that
given a higher-order functional program \(P\) and a formula \(A\)
of the modal \(\mu\)-calculus (or, equivalently, an alternating parity tree automaton),
one can effectively construct an \hflz{} formula \(\form_{P,A}\) such that
the program \(P\) satisfies the property described by \(A\), if and only if
the formula \(\form_{P,A}\) is valid. Various temporal property verification problems
(including safety, termination, CTL, LTL, \ctlstar{} verification)
can thus be reduced to the \hflz{} validity checking problem and solved in a uniform manner.
Here, we just give some examples, instead of reviewing the general reduction.

Let us consider the following OCaml program.
\begin{verbatim}
let rec fib x k = if x<2 then k x else fib (x-1) (fun y -> fib (x-2) (fun z -> k(y+z))
\end{verbatim}
The function \texttt{fib} computes the Fibonacci number in the continuation-passing style.
The termination of \texttt{fib x (fun r->())}\footnote{Here, we restrict the answer type
  of the continuation to the type \texttt{unit} of the unit value
  \texttt{()}, which is mapped to the type
  \(\Prop\) of propositions by the translation to \hflz{} validity checking.}
for all \texttt{x} can be reduced
to the validity of \(\forall x.{\fib}\,x\,(\lambda r.\TRUE)\), where the predicate
\({\fib}\) is defined by:
\[
\begin{array}{l}
\fib^{\INT\to(\INT\to\Prop)\to\Prop}\, x\,k =_\mu 
(x<2 \imp k\,x)\land (x\ge 2\imp \fib\,(x-1)\,\lambda y.\fib\,(x-2)\,\lambda z.k(y+z)).
\end{array}
\]
Here, \(b\imp \form\) abbreviates \(\neg b\lor \form\).
The equation above defines \(\fib\) as the least predicate that
satisfies the equation (recall Notation~\ref{not:eq}).
The formula mimics the structure of
the program. In particular,
the parts ``\(x<2 \imp k\,x\)'' and
``\(x\ge 2\imp \cdots\)'' respectively correspond to the then-part and the else-part
of the function definition of \texttt{fib}.

The property ``\texttt{fib x (fun r->assert(r>=x))} never fails for any \texttt{x}''
can be expressed by \(\forall x.\fibsafe\, x\,(\lambda r.r\ge x)\), where \(\fibsafe\) is defined by:
\[
\begin{array}{l}
\fibsafe^{\INT\to(\INT\to\Prop)\to\Prop}\, x\,k =_\nu 
(x<2 \imp k\,x)\land (x\ge 2\imp \fibsafe\,(x-1)\,\lambda y.\fibsafe\,(x-2)\,\lambda z.k(y+z)).
\end{array}
\]
As in the examples above, (i) the \hflz{} formula obtained by the reduction mimics
the structure of the original program, and (ii) liveness properties (like termination)
are expressed by using the least fixpoint operator \(\mu\), and safety properties
(like partial correctness) are
expressed by using the greatest fixpoint operator \(\nu\).

For the automated verification of first-order programs, it has been
a popular approach to reduce verification problems
to the satisfiability problem for Constrained Horn Clauses (CHC)~\cite{Bjorner15}.
Since the satisfiability problem for CHC (where data domains are restricted to integers) can be reduced to the validity checking
problem for the first-order \nuhflz{}~\cite{DBLP:conf/sas/0001NIU19},
the program verification framework based on \hflz{} can be considered an extension
of the CHC-based program verification framework with higher-order
features and fixpoint alternations.
\hflz{} can also be viewed as an extension of HoCHC (higher-order CHC)~\cite{DBLP:journals/pacmpl/BurnOR18} with
fixpoint alternations (recall \tabref{tab:fixpoint-logics}).

\section{Reduction from \hflz{} to \nuhflz{}}
\label{sec:method}
\subsection{Overview of \hflz{} Validity Checking}

Fig.~\ref{fig:hflz} shows the overall flow of our \hflz{} validity checking
method. Given a \hflz{} formula \(\form\),
we approximate \(\form\) by a \nuhflz{} formula \(\form'\), by removing
all the least fixpoint formulas (of the form \(\mu x.\psi\)).
The formula \(\form'\) is an \emph{under}-approximation of
\(\form\), in the sense that if \(\form'\) is valid, then so is \(\form\).
We then check whether \(\form'\) is valid by using an existing validity checker
for \nuhflz{}~\cite{DBLP:conf/sas/IwayamaKST20,DBLP:conf/aplas/KatsuraIKT20}.
If \(\form'\) is valid, then we can conclude that \(\form\) is also valid.
Otherwise, we refine the approximation of \(\form\) and repeat the cycle.
As the procedure in Fig.~\ref{fig:hflz} can only conclude the validity of
a given formula, we actually run the procedure for a given formula \(\form\)
and its negation \(\neg\form\) (which can also be represented as a \hflz{} formula,
by taking the dual of each operator) in parallel. If the procedure for \(\neg\form\)
returns ``valid'', then we can conclude that \(\form\) is invalid.
Since \hflz{} validity checking is undecidable in general,
the whole procedure is of course sound but incomplete: for some input,
the procedure may repeat the cycle indefinitely, or the backend \nuhflz{} validity
checker~\cite{DBLP:conf/sas/IwayamaKST20,DBLP:conf/aplas/KatsuraIKT20}
may not terminate.

The main technical issue in the procedure sketched above is how to approximate
\(\mu\)-formulas, which is the focus of the rest of
this section. We first discuss a basic method in Section~\ref{sec:basic},
and then discuss how to improve the precision of the approximation by adding
extra arguments for higher-order predicates in Section~\ref{sec:ho}.
We then further improve the approximation by removing redundant extra arguments
in Section~\ref{sec:opt}.
\begin{figure}[tp]
  \begin{center}
    \includegraphics[scale=0.5]{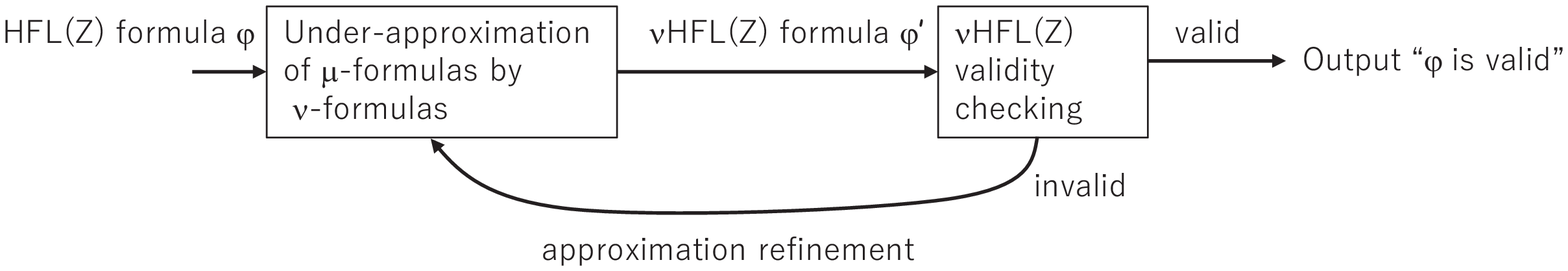}
  \end{center}
  \caption{Overall flow of \hflz{} validity checking.}
  \label{fig:hflz}
\end{figure}

\subsection{Basic Method}
\label{sec:basic}

As mentioned in Section~\ref{sec:intro}, the basic idea is to
under-approximate each \(\mu\)-formula \(\mu x^{\fty}.\form(x)\) (where
\(\form\) has type \(\fty\to\fty\)) by 
\(\form^n (\bot_\fty) \defeq \underbrace{\form(\cdots(\form(}_n \bot_\fty))\cdots)\).
Here, \(\bot_\fty\) is the least formula of type \(\fty\), defined by:
\(\bot_{\Prop} = \FALSE\) and \(\bot_{\sty\to\fty} = \lambda y^\sty.\bot_{\fty}\).
To see that
\(\form^n (\bot_\fty)\) is an underapproximation of \(\mu x^{\fty}.\form(x)\),
recall that \(\sem{\form}\) is a monotonic function.
Thus,
we have:
\begin{align*}
  \sem{\form^n (\bot_\fty)}&= \sem{\form}^n(\Bot_\fty)\\
  &\LEQ_{\fty}  \sem{\form}^n(\LFP_{\fty}(\sem{\form})) \quad \mbox{(by
    \(\Bot_\fty \LEQ_{\fty}\LFP(\sem{\form})\) and the monotonicity of \(\sem{\form}\))}\\
  &= \LFP_{\fty}(\sem{\form}) \quad \mbox{(by the definition of \(\LFP_{\fty}\))}\\
  &= \sem{\mu x^{\fty}.\form(x)}.
\end{align*}

Since the appropriate number \(n\) may depend on the values of free variables
in \(\form\), we actually use the \(\nu\)-formula \(\psi\,n\) to
represent \(\form^n (\bot_\fty)\), where \(\psi\) is:
\[
\nu x^{\INT\to\fty}.\lambda \cv^{\INT}.\lambda \seq{y}^{\seq{\sty}}.\cv>0 \land
\form\,(x\,(\cv-1))\,\seq{y}
\]
Here, \(\seq{y}^{\seq{\sty}}\) denotes a sequence \(y_1^{\sty_1},\ldots,y_\ell^{\sty_\ell}\)
and we assume that \(\fty = \sty_1\to\cdots \to\sty_\ell\to\Prop\).
Note that the original formula
\(\mu x^{\fty}.\form\,x\) is \(\eta\)-equivalent to
\(\mu x^{\fty}.\lambda \seq{y}^{\seq{\sty}}.\form\,x\,\seq{y}\); thus,
the main differences of the formula \(\psi\,n\) from the original formula are:
(i) an extra integer argument \(u\) has been added to \(x\), to count the number of
iterations \(n\), and (ii) the least fixpoint operator \(\mu\) has been replaced by
the greatest fixpoint operator \(\nu\).
We can confirm that the formula \(\psi\,n\) is equivalent to 
\(\form^n (\bot_\fty)\) as follows.
For \(n=0\), we have:
\begin{align*}
&  \psi\,n
  \equiv
(\lambda \cv^{\INT}.\lambda \seq{y}^{\seq{\sty}}.\cv>0 \land
  \form\,(\psi\,(\cv-1))\,\seq{y})\, 0 
  \equiv
\lambda \seq{y}^{\seq{\sty}}.0>0 \land
\form\,(\psi\,(-1))\,\seq{y} 
\equiv
\lambda \seq{y}^{\seq{\sty}}.\FALSE 
\equiv\bot_{\fty}.
\end{align*}
and for \(n>0\),
\[
\begin{array}{rcll}
  \psi\,n
  &\equiv&
(\lambda \cv^{\INT}.\lambda \seq{y}^{\seq{\sty}}.\cv>0 \land
  \form\,(\psi\,(\cv-1))\,\seq{y})\, n \qquad& \mbox{(unfolding \(\nu\))}\\
  &\equiv&
\lambda \seq{y}^{\seq{\sty}}.n>0 \land
  \form\,(\psi\,(n-1))\,\seq{y} & \mbox{($\beta$-equality)}\\
  &\equiv&
\lambda \seq{y}^{\seq{\sty}}.
  \form\,(\psi\,(n-1))\,\seq{y}& \mbox{(by assumption \(n>0\))}\\
  &\equiv&
\lambda \seq{y}^{\seq{\sty}}.
  \form\,(\form^{n-1}(\bot_{\fty}))\,\seq{y}& \mbox{(by induction on \(n\))}\\
  &\equiv&
  \form^{n}(\bot_{\fty})& \mbox{(by $\eta$-equality)}
\end{array}
\]
When \(\form\) contains free integer variables \(w_1,\ldots,w_m\), 
the number \(n\) needed to properly approximate the original formula
may depend on them. Thus, we actually replace the number \(n\) with
an expression \(c|w_1|+\cdots + c|w_m|+d\), where \(c\) and \(d\) are
some non-negative integers.
Due to the monotonicity of \(\psi\,n\) with respect to \(n\), 
we can improve the precision by increasing the values of \(c\) and \(d\).

The method sketched above is a generalization of Kobayashi et al.'s method~\cite{DBLP:conf/sas/0001NIU19}
for
the first-order fragment of \hflz{} (called Mu-Arithmetic) to full \hflz{}
(which was in turn a generalization
 the method of Fedyukovich et al.~\cite{freqterm} for termination analysis).
We give a few examples below.
Example~\ref{ex:partial-app} highlights a subtle issue caused by the generalization
to the higher-order case.
  
\begin{example}
\label{ex:simple-revisited}  
  Recall the formula \(\form_1\defeq\mu x^{\INT\to\Prop}.\lambda y.y=0\lor x(y-1)\)
  in Example~\ref{ex:simple}.
  Suppose we wish to prove the validity of
  \(\forall w.w<0 \lor \form_1\,w\).
  Based on the method sketched above, we approximate \(\form_1\) with
  \(
\psi\,(c|y|+d)
\),
where \(\psi\) is:
\[
\nu x^{\INT\to\INT\to\Prop}.\lambda \cv^{\INT}.\lambda y^{\INT}.\cv>0 \land
(y=0\lor x(\cv-1)(y-1)).
\]
The resulting formula
\(\forall w.w<0 \lor \psi\,(c|w|+d)\,w\)
can automatically be proved valid for \(c=d=1\), by using an existing \nuhflz{} validity checker
like \rethfl{}~\cite{DBLP:conf/aplas/KatsuraIKT20}.

To confirm the validity of
\(\forall w.w<0 \lor \psi\,(|w|+1)\,w\) \emph{manually},
it suffices to observe that
(the semantics of) \(\lambda \cv.\lambda y.y\ge 0\land \cv>y\) is a post-fixpoint of
\(\form_0 \defeq \lambda x^{\INT\to\INT\to\Prop}.\lambda \cv^{\INT}.\lambda y^{\INT}.\cv>0 \land
(y=0\lor x(\cv-1)(y-1))\).
Indeed, we have:
\begin{align*}
  &  \form_0(\lambda \cv.\lambda y.y\ge 0\land \cv>y)\\
  &\equiv \lambda \cv.\lambda y.\cv>0
    \land (y=0\lor (\lambda \cv.\lambda y.y\ge 0\land \cv>y)(\cv-1)(y-1))\\
  &\equiv\lambda \cv.\lambda y.\cv>0
    \land (y=0\lor (y-1\ge 0\land \cv-1>y-1)\\
  &\equiv\lambda \cv.\lambda y.(\cv>0\land y=0)\lor (y\ge 1\land \cv>y)\\
  &\equiv\lambda \cv.\lambda y.y\ge 0\land \cv>y.
\end{align*}
Thus, \(\lambda \cv.\lambda y.y\ge 0\land \cv>y\) is a post-fixpoint
(and actually also a fixpoint)
of \(\form_0\);
hence, we have
\[\sem{\lambda \cv.\lambda y.y\ge 0\land \cv>y} \LEQ \GFP(\form_0)=\sem{\psi}.\]
Therefore, we have:
\begin{align*}
    \sem{\forall w.w<0 \lor \psi\,(|w|+1)\,w} 
 & \sqsupseteq
  \sem{\forall w.w<0 \lor (\lambda \cv.\lambda y.y\ge 0\land \cv>y)\,(|w|+1)\,w}\\
  & =
  \sem{\forall w.w<0 \lor (w\ge 0\land |w|+1>w)} 
  = \Top.
\end{align*}
\end{example}
  \begin{remark}
    The transformation sketched above passes around extra arguments of the form
    \(c|x_1|+\cdots +c|x_k|+d\).
    The use of the absolute value operator is actually problematic for backend solvers for
    \nuhflz{}.
    Thus, in the actual implementation,
    we replace each formula of the form \(\form\,(c|x_1|+\cdots +c|x_k|+d)\)
    with
    \[\lambda \seq{y}.\forall u.(\bigwedge_{s_1,\ldots,s_k\in\set{-1,1}} u\ge cs_1x_1+\cdots+cs_kx_k+d)
    \imp \form\,u\,\seq{y}.\]
    For example, the formula
    \(\lambda y.\psi\,(c|y|+d)\,y\) above is actually replaced by:
   \[ \lambda y.\forall u.u\ge cy+d\land u\ge -cy+d \imp \psi\,u\,y.\]
   Note that due to the monotonicity of \(\psi\,u\,y\) with respect to \(u\),
   the replacement does not change the semantics of formulas. \qed
   \end{remark}

  \begin{example}
    \label{ex:simple-var}
    Suppose that we wish to prove the validity of
  \(\forall w.w<0 \lor \form_1\,(2w)\),
  instead of the formula   \(\forall w.w<0 \lor \form_1\,w\),
  in Example~\ref{ex:simple-revisited}.
  In this case, the approximate formula
  \(\forall w.w<0 \lor \psi\,(c|w|+d)\,(2w)\)
  is \emph{invalid} when \(c=d=1\). In fact, for \(w=1\),
  \begin{align*}
    \psi\,(|w|+1)\,(2w)  & \equiv
    \psi\,2\,2 \equiv \psi\,1\,1\equiv \psi\,0\,0\equiv \FALSE.
  \end{align*}
  In such a case, we proceed to the approximation refinement step in
  Fig.~\ref{fig:hflz} to increase the values of \(c\) and \(d\).
  By increasing the values of \(c\) and \(d\) to \(2\),
  we obtain a better approximation:
  \(\forall w.w<0 \lor \psi\,(2|w|+2)\,(2w)\), which can be proved valid.
  \qed
\end{example}
  
\begin{example}
  Recall the formula \(\form_2\defeq \forall x.\fib\,x\,(\lambda r.\TRUE)\) in Section~\ref{sec:app}
  (which was obtained by encoding the termination problem for Fibonacci function),
  where \(\fib\) is defined by:
  \[
\begin{array}{l}
\fib^{\INT\to(\INT\to\Prop)\to\Prop}\, x\,k =_\mu 
(x<2 \imp k\,x)\land (x\ge 2\imp \fib\,(x-1)\,\lambda y.\fib\,(x-2)\,\lambda z.k(y+z)).
\end{array}
\]
The formula \(\form_2\) can be approximated by
\(\forall x.\fib'\,(c|x|+d)\,x\,(\lambda r.\TRUE)\),
where \(\fib'^{\INT\to\INT\to(\INT\to\Prop)\to\Prop}\) is defined by:
\[
\begin{array}{l}
\fib'\,\cv\,x\,k =_\mu 
\cv>0\,\land 
(x<2 \imp k\,x)\,\land\\\qquad \qquad\qquad \qquad
(x\ge 2\imp \fib'\,(\cv-1)\,(x-1)\,\lambda y.\fib'\,(\cv-1)\,(x-2)\,\lambda z.k(y+z)).
\end{array}
\]
The resulting formula can be proved valid for \(c= d=1\). \qed
\end{example}

The following example involves a partial application of a predicate defined by \(\mu\).
\begin{example}
  \label{ex:partial-app}
  Consider the formula \(\forall x.x\ge 0\imp G\, (F\, x)\, 0\) where \(G\) and \(F\) are defined by:
  \[
  \begin{array}{l}
    G\,f\,y\,=_\nu f\,y \land G\,f\,(y+1)\qquad
    F\,x\,y\,=_\mu x+y\le 0 \lor F\,(x-1)\,y.
  \end{array}
  \]
  Here, to approximate the formula \(F\), the number of unfoldings should be at least
  \(x+y+1\), but the value of \(y\) is not available in the partial application \(F\,x\).
  To remedy the problem, it suffices to \(\eta\)-expand \(F\,x\) and replace the main
  formula with \(\forall x.x\ge 0\imp G (\lambda y.F\, x\,y) 0\).
  We can then apply the approximation as sketched above, and obtain
  \(\forall x.x\ge 0\imp G\,(\lambda y.F\,(c_1|x|+c_2|y|+d)\,x\,y)\,0\), where:
  \[
  \begin{array}{l}
    G\,f\,y\,=_\nu f\,y \land G\,f\,(y+1)\qquad
    F\,u\,x\,y\,=_\nu (u>0)\land (x+y\le 0 \lor F\,(u-1)\,(x-1)\,y).
  \end{array}
  \]
  The resulting formula can be proved valid for \(c_1=c_2=d=1\).
  In this manner, each partial application of a least-fixpoint predicate 
  is \(\eta\)-expanded before the predicate is approximated by a greatest-fixpoint predicate.
  \qed
\end{example}

\begin{remark}
  \label{rem:ack}
\newcommand\Ack{\textit{Ack}}
The method above is not sufficient, for example, for proving the termination of
Ackermann function:
\begin{verbatim}
  let rec ack y z k = if y=0 then k(z+1) else if z=0 then ack (y-1) 1 k
                      else ack y (z-1) (fun x-> ack (y-1) x k).
\end{verbatim}
The termination of \texttt{ack m n (fun r->())} is expressed by:
\(\Ack\;m\,n\,\lambda r.\TRUE\), where
\[
\begin{array}{ll}
  \Ack\,y\,z\,k =_\mu
  & (y=0\imp k(z+1)) 
  \land (y\ne 0\land z=0\imp \Ack\,(y-1)\,1\,k)\\
  & \land (y\ne 0\land z\ne 0\imp \Ack\,y\,(z-1)\,(\lambda x.\Ack\,(y-1)\,x\,k)).
\end{array}
\]
Adding a single parameter \(u\) to count the number of unfoldings
\[
\begin{array}{l}
  \Ack'\;u\,y\,z\,k =_\nu 
  u>0\land \big( 
   (y=0\imp k(z+1)) 
  \ \land (y\ne 0\land z=0\imp \Ack'\,(u-1)\,(y-1)\,1\,k)\\\qquad\qquad\qquad\qquad
  \ \land (y\ne 0\land z\ne 0\imp \Ack'\,(u-1)\,y\,(z-1)\,(\lambda x.\Ack'\,(u-1)\,(y-1)\,x\,k))\big)
\end{array}
\]
does not work, since the depth of the recursive calls of \texttt{ack y z k} is not linear
in \texttt{y} and \texttt{z}.
As suggested by Kobayashi et al.~\cite{DBLP:conf/sas/0001NIU19} for the first-order case,
to deal with the example above, we need to prepare \emph{two} counters
and approximate
\(\Ack\;m\,n\,(\lambda r.\TRUE)\)
by
\(\Ack'\;(c|m|+c|n|+d)\;(c|m|+c|n|+d)\;m\,n\,(\lambda r.\TRUE)\), where:
\[
\begin{array}{l}
  \Ack'\;u_1\;u_2\,y\,z\,k =_\nu \\\quad
  u_1>0\land u_2>0\land \big( 
   (y=0\imp k(z+1)) 
  \ \land (y\ne 0\land z=0\imp \Ack''\,u_1\,u_2\,(y-1)\,1\,k)\\\qquad\qquad\qquad\qquad
  \ \land (y\ne 0\land z\ne 0\imp \Ack''\,u_1\,u_2\,y\,(z-1)\,(\lambda x.\Ack''\,u_1\,u_2\,(y-1)\,x\,k))\big)\\
  \Ack''\;u_1\;u_2\,y\,z\,k =_\nu 
  \Ack'\;u_1\;(u_2-1)\;y\,z\,k   \hfill \mbox{(decrement \(u_2\), or)}\\  \qquad\qquad\qquad\qquad
   \lor   \Ack'\;(u_1-1)\;(c|y|+c|z|+d)\;y\,z\,k\hfill \mbox{(decrement \(u_1\), and reset \(u_2\))}.
\end{array}
\]
In general, given a formula \(X\,\seq{v}\) where \(X\) is defined by:
\( X\,\seq{y} =_\mu \form(X)\),
we can approximate it with
\(\forall u_{k-1},\ldots,
 u_0\ge c( u_j\Sigma |v_i|)+d.X_{\MC}\,u_{k-1}\,\cdots\,u_0\,\seq{v}\),
where:
\begin{align*}
 & X_{\MC}\,u_{k-1}\,\cdots\,u_0\,\seq{y} =_\nu \\ & \quad
  (u_{k-1}\ge 0\land \cdots \land u_0\ge 0)\land\\ & \quad
  \form(\lambda \seq{y}'. 
  \forall u'_{k-2},\ldots,u'_0\ge
  c(\Sigma_{0\le j< k} u_j + \Sigma |y'_i|)+d. 
  X_{\MC}\,(u_{k-1}-1)\,u_{k-2}'\,\cdots\,u'_0\,\seq{y}'\\  &\qquad\quad
  \lor\forall u'_{k-3},\ldots,u'_0\ge
  c(\Sigma_{0\le j< k} u_j + \Sigma |y'_i|)+d.X_{\MC}\,u_{k-1}\,(u_{k-2}-1)\,u'_{k-3}\,\cdots\,u'_0\,\seq{y'}\\&\qquad\quad
  \lor \cdots \lor
  X_{\MC}\,u_{k-1}\,\cdots\,u_1\,(u_0-1)\,\seq{y}')
\end{align*}
Here, the notation \(\forall u_{k-i},\ldots,u_1\ge e.\form\) abbreviates
\(\forall u_{k-i},\ldots,u_1. u_{k-i}\ge e\land \cdots \land u_{1} \ge e
\imp \form\); in particular, \(e\) is the lower-bound for
all the variables \(u_{k-i},\ldots,u_1\).
As we discuss in Section~\ref{sec:disc}, 
the basic method with a single counter \(u\) is analogous
to (but strictly more powerful than)
termination verification using single linear ranking functions,
and the extension with multiple counters is
strictly more powerful than
methods based on
lexicographic linear ranking functions~\cite{DBLP:conf/tacas/CookSZ13,DBLP:conf/cav/BradleyMS05}
and disjunctive well-founded relations based on linear ranking functions.
Since the extension with multiple counters is orthogonal to the extensions discussed below,
we focus on the method using a single counter below.
\qed
\end{remark}

\begin{example}
  Let us consider the formula:
  \[\big(\nu f.\lambda x.(\mu g.\lambda y.(y=0\land f(x+1))\lor (y\ne 0\land g\,(y-1)))x\big)0,\]  whose ``alternation depth'' (as defined for the modal \(\mu\)-calculus)~\cite{DBLP:reference/mc/BradfieldW18} is 2.
  The formula corresponds to the property that the function \(f\) is called
  infinitely often in the following OCaml-like program:
\begin{verbatim}
  let rec f x = 
    let rec g y = if y=0 then f(x+1) else g(y-1) in g x
  in f 0.
\end{verbatim}
By our approximation,  we obtain \(F\,0\), where:
\begin{align*}
  & F\,x =_\nu G\,(c|x|+d)\,x\,x\qquad
   G\,u\,x\,y =_\nu u>0\land ((y=0\land F\,(x+1))\lor (y\ne 0\land G\,(u-1)\,x\,(y-1))
\end{align*}
(the second parameter \(x\) of \(G\) is introduced by lambda lifting), which can be proved valid for \(c=d=1\). \qed
\end{example}

\begin{remark}
  \label{rem:computability}
  From the viewpoint of computability theory, our approach of reducing \hflz{} validity checking
  to \nuhflz{} validity checking has the following fundamental limitation.
  The \hflz{} validity checking problem is \(\Pi^1_1\)-hard and
  \(\Sigma^1_1\)-hard, since the fair termination problem (which is \(\Pi^1_1\)-complete~\cite{DBLP:journals/jacm/Harel86})
  and its dual can be reduced to \hflz{} validity checking; see also \cite{DBLP:conf/lics/Tsukada20}.
  In contrast, (the validity checking problem for) the \(\nuhflz{}\)- formula obtained by our reduction
  belongs to \(\Pi^0_1\)
  (in other words,
  the set of valid \nuhflz{} formulas is co-recursively enumerable),
  since
  the validity of a \nuhflz{} formula can be \emph{disproved} by unfolding greatest fixpoint formulas
  a finite number of times and showing the resulting formula is invalid.
  This implies that there is no complete, effective procedure to reduce \hflz{} validity checking to
  \nuhflz{} validity checking.
  Despite this theoretical limitation, however, as reported in Section~\ref{sec:exp}
  (where the benchmark set includes instances of the fair termination problem),
  our method can solve many instances of the \hflz{} validity checking problem that have been obtained
  from actual program verification problems.
  This kind of phenomenon has often been observed in the context of automated program verification: 
  the fair termination verification problem (which is \(\Pi^1_1\)-complete) has been solved
  by a reduction to the safety property
  verification problem (which is \(\Pi^0_1\)-complete, hence much easier in theory) in \cite{Cook07POPL,MTSUK16POPL}. See also the discussion in Section~\ref{sec:disc}.
  \qed
\end{remark}

\begin{remark}
  \label{rem:complete-fragment}
  A reader may expect a syntactic characterization of the class of \hflz{} formulas
  for which our method is complete with respect to the hypothetical completeness
  of the backend solver for \nuhflz{} validity checking.
  Our method is indeed complete for the \(\mu\)-only fragment of \hflz{} (i.e.,
  the fragment without the greatest fixpoint operators), in that given a
  closed valid \(\mu\)-only formula \(\form\), our procedure eventually terminates
  and concludes that the formula is valid (see \cite{ESOP2018full}, Lemma~6).
  It seems difficult to give
  a clear syntactic characterization of a larger, more useful class of \hflz{} formulas
  for which our method is complete. As discussed in Example~\ref{ex:quantifiers},
  universal and existential quantifiers \(\forall x.p(x)\) and
  \(\exists x.p(x)\) can be expressed by
  \(\Forall\;p\) and \(\Exists\;p\) respectively, where:
  \begin{align*}
&    \Forall\;p =_\nu p\,0 \land \Forall\;(\lambda x.p(x-1))\land  \Forall\;(\lambda x.p(x+1))\\
&    \Exists\;p =_\mu p\,0 \lor \Exists\;(\lambda x.p(x-1))\lor \Exists\;(\lambda x.p(x+1))\\
  \end{align*}
  By passing around the predicates \(\Forall\) and \(\Exists\) above
  through higher-order predicates, 
   one can express arbitrary nesting of quantifiers to realize
   any \(\Sigma^0_n\) and \(\Pi^0_n\) formulas (for any \(n\))
   without any \emph{syntactic}
   nesting of greatest and least fixpoint operators. In contrast,
   as mentioned in Remark~\ref{rem:computability},
   the \(\nuhflz{}\)- formula obtained by our reduction
   belongs to \(\Pi^0_1\).
   Instead of trying to give a syntactic characterization,
   in Section~\ref{sec:disc}, we compare the class of formulas for
   which our method is complete with those for which previous representative methods
   are complete, and show that the former is strictly larger than the latter.
\qed
\end{remark}
\subsection{Adding Extra Arguments for Higher-Order Values}
\label{sec:ho}
\newcommand\All{\mathit{All}}
\newcommand\Succ{\mathit{Succ}}
\newcommand\Pred{\mathit{Pred}}
To deal with higher-order predicates, we need to extend the basic method
to take function arguments into account.
We explain the method and our solution through an example.
 Let us consider the formula \(\All\;(\lambda k.k\,0)\), where \(\All\) and \(F\) are defined by:
  \[
  \begin{array}{l}
    \All 
    \; x^{(\INT\to\Prop)\to\Prop} =_\nu
    F\,x \land \All\;(\Succ\, x)\qquad
    F\;x^{(\INT\to\Prop)\to\Prop} =_\mu
    x(\lambda y.y=0) \lor F\;(\Pred\,x)\\
    \Succ\;x^{(\INT\to\Prop)\to\Prop}\;k =_\nu x(\lambda y.k(y+1))\qquad
    \Pred\;x^{(\INT\to\Prop)\to\Prop}\;k =_\nu x(\lambda y.k(y-1)).
  \end{array}
  \]
  This is a higher-order variant of \(\forall y.y\ge 0\imp 
  (\mu f^{\INT\to\Prop}.\lambda y.y=0\lor f(y-1))y\) considered in Example~\ref{ex:simple},
  where an integer \(y\) has been replaced by a higher-order-predicate
  \(\lambda k.k\,y\) of type \((\INT\to\Prop)\to\Prop\).

  Since \(F\) is defined by \(\mu\), we remove it by approximating it with \(\nu\).
  The basic translation in Section~\ref{sec:basic} would yield:
  \[
  \begin{array}{l}
    \All' 
    \; x^{(\INT\to\Prop)\to\Prop} =_\nu
    F'\,d\,x \land \All'\;(\Succ\, x)\\
    F'\;u\;x^{(\INT\to\Prop)\to\Prop} =_\mu
    u>0\land (x(\lambda y.y=0) \lor F'\;(u-1)\;(\Pred\,x)).\\
  \end{array}
  \]
   (We have omitted the definitions of \(\Succ\) and \(\Pred\) as they are unchanged.)
  Here, the argument \(u\) of \(F'\)
  represents the number of unfoldings for the original predicate \(F\),
  and the bound \(d\) for \(u\) is a \emph{constant}.
  This is because there is no integer variable in the scope of the body of \(\All\).
  Since the value of \(u\) should actually be greater than the value represented by \(x\),
  the formula \(\All'\,(\lambda k.k\,0)\)
  is invalid; thus, we fail to prove the validity of
   the original formula \(\All\,(\lambda k.k\,0)\).

   To remedy the problem above, we add to \(\All'\) an extra integer argument that
   represents information about \(x\). We thus refine the approximation of
   \(\All\;(\lambda k.k\,0)\) to
   \(\All''\;(d', \lambda k.k\,0)\), where \(\All'\) is defined by:
  \[
  \begin{array}{l}
    \All''\; (v_x^{\INT},  x^{(\INT\to\Prop)\to\Prop}) =_\nu
    F'\,(cv_x+d)\,x \land \All''\;(c'v_x+d', \Succ\, x)
  \end{array}
  \]
  For technical convenience, we have extended the syntax of formulas with pairs
  (which can be removed by the standard currying transformation). The new argument \(v_x\)
  of \(\All''\) carries information about \(x\),
  which is updated to \(c'v_x+d'\) upon a recursive call of \(\All'\).
  The predicate \(F'\) remains the same, and \(c',d'\) are some positive integer constants.
  The formula \(\All''\;(d', \lambda k.k\,0)\) can now be proved valid
  for \(c=d=c'=d'=1\).

  The idea of adding extra arguments above has been inspired by
  Unno et al.'s method of adding extra arguments for relatively complete
  refinement type inference~\cite{UnnoTK13}. Unlike in the case of
  Unno et al.'s method~\cite{UnnoTK13}, however, our method above satisfies
  the monotonicity property on extra arguments. Because the extra arguments are
  used only for computing the lower-bound of the number of unfoldings of \(\mu\)-formulas,
  the precision of the approximation monotonically increases with respect to
  the values of coefficients \(c',d'\) (see Theorem~\ref{th:monotonicity} given later). Thus, like the values of \(c,d\),
  we just need to monotonically increase the values of \(c',d'\) to refine the precision
  of approximation. In contrast, in Unno et al.'s method, 
  a rather complex procedure is required
  to infer appropriate extra arguments in a counterexample-guided manner.

  \subsection{Optimization Transformation}
\label{sec:opt}
\label{sec:methodOptimization}
  A remaining issue is how to decide where 
  we should insert extra integer arguments.
  A naive way would be to add an extra integer argument to \emph{every}
  function argument, but then too many arguments would be introduced,
  causing a burden for the backend validity checker for \nuhflz{}.
  For example, for the above example, the naive approach would yield:
  \[
  \begin{array}{l}
    \All 
    \;(v_x, x) =_\nu 
   (F\,(cv_x+d)\,(c'v_x+d',x)) \land (\All\;(c'v_x+d', \Succ\, (c'v_x+d',x)))\\
    F\;u\;(v_x,x) =_\nu 
    (u>0)\land (x(c'v_x+d', \lambda y.y=0) \lor F\;(u-1)\;(c'v_x+d', \Pred\, (c'v_x+d',x)))\\
    \Succ\;(v_x,x)\;(v_k, k) =_\nu x\,(c'v_x+c'v_k+d',\lambda y.k(y+1))\\
    \Pred\;(v_x,x)\;(v_k, k) =_\nu x\,(c'v_x+c'v_k+d', \lambda y.k(y-1)).
  \end{array}
  \]
  The extra arguments \(v_x, v_k\) for \(F,\Succ\), and \(\Pred\) would however be redundant,
  because they do not flow to the argument \(u\) of \(F\).

  We introduce below types for representing
  where extra arguments should be inserted, and
  formalize the translation from \hflz{} formulas
  to \nuhflz{} formulas as a type-based transformation.

  We first extend types with \emph{tags}.
  The sets of \emph{tagged argument types} and
  \emph{tagged (predicate) types}, ranged over by \(\alpha\) and \(\efty\) respectively,
  are given by:
  \[
  \begin{array}{l}
  \tagty \mbox{ (tagged argument types) } ::= \INT \mid (\efty, \tagv)\\
  \efty \mbox{ (tagged predicate types)} ::= \Prop \mid \tagty\to\efty\qquad
  \tagv \mbox{ (tags) } ::= \taguse\mid \tagnotuse.\\
  \end{array}
  \]
  A tagged type \((\efty,\taguse)\)
  represents the type of a predicate argument for which an extra integer argument
  is required for an approximation of some \(\mu\)-formula.
  For example, recall the predicate \(\All\) in Section~\ref{sec:ho}:
    \[
  \begin{array}{l}
    \All 
    \; x^{(\INT\to\Prop)\to\Prop} =_\nu
    F\,x \land \All\;(\Succ\, x)\qquad
    F\;x^{(\INT\to\Prop)\to\Prop} =_\mu \cdots.
  \end{array}
  \]
The argument \(x\) of \(\All\) should have type
  \(((\INT\to\Prop,\tagnotuse)\to\Prop, \taguse)\),
  because information about \(x\) is required to estimate how often \(F\) should
  be unfolded, whereas information about the argument of \(x\) is not.

  We formalize the (optimized) transformation from
  \hflz{} formulas to \nuhflz{} formulas as
  a type-based transformation relation
  \(\judgesimp{\envv}{\rho}{\form}{\efty}{\form'}\)
  where: (i) \(\envv\),
  called an tagged type environment,
  is a finite map from variables to tagged argument types, 
  (ii) \(\form\) and \(\form'\) are the input and output of the transformation.
  We also use an auxiliary transformation relation
  \(\judgesimp{\envv}{\rho}{\form}{\tagty}{\form'}\) for the translation of
  an argument.

  We need to introduce some notations to define the transformation relation.
  For a type environment \(\envv\) and
  a set \(V\) of variables, we write \(\restrict{\envv}{V}\) for the restriction of
  \(\envv\) to \(V\), i.e., \(\set{x\COL\tagty\in\Delta\mid x\in V}\).
  We write \(\gettags(\Delta)\) for the set of outermost tags in \(\Delta\), defined by:
  \[
    \gettags(\Delta) = \set{\tagv\mid x\COL (\efty,\tagv)\in\Delta}.
    \]
    For tagged argument types \(\tagty\) and \(\tagty'\), we write
    \(\tagty\raweq\tagty'\) when they are identical except their outermost tags,
    i.e., if either \(\tagty=\tagty'=\INT\), or \(\tagty=(\efty,\tagv)\) and \(\tagty=(\efty,\tagv')\)
    for some \(\efty, \tagv\), and \(\tagv'\).
    For tagged types \(\tagty\) and \(\efty\), we write
    \(\ST(\tagty)\) and \(\ST(\efty)\) for the simple types obtained by removing
    the tags. We also write \(\ST(\Delta)\) for the simple type environment
    defined by \(\ST(x_1\COL\tagty_1,\ldots,x_k\COL\tagty_k)=
    x_1\COL\ST(\tagty_1),\ldots,x_k\COL\ST(\tagty_k)\).
  We sometimes write \(\letexp{x}{\form_1}{\form_2}\) for \((\lambda x.\form_2)\form_1\).
  
  The transformation relations are defined by the rules in \figref{fig:trans}.
  The first two rules are for the translation of arguments.
  As specified in \rn{Tr-TagT}, if the tag is \(\taguse\), then we add an extra
  argument \(\exarg(\restrict{\envv}{\FV(\form)})\),
  where \(\exarg\) is defined by:\footnote{For the sake of simplicity, we do not distinguish between the coefficients
    \(c,d\) for estimating the number of unfoldings of \(\mu\)-formulas,
    and \(c',d'\) for computing extra arguments. The actual implementation reported in
  Section~\ref{sec:exp} distinguishes between \(c,d\) and \(c',d'\).}

    \[\exarg(\envv)=d+c(\Sigma_{x\COL\INT\in \envv}  |x| +
    \Sigma_{x\COL(\efty,\taguse)\in\envv} \exv{x}).\]
  It is a linear combination of (the absolute value of) original integer variables \(x\)
  and auxiliary integer variables \(\exv{x}\).
  We fix the name of the auxiliary integer variable associated with \(x\)
  to \(\exv{x}\), and assume that it does not clash with the names of other variables.
  We ensure that \(\exv{x}\) always takes a non-negative integer value, so that we need
   not take the absolute of \(\exv{x}\) in \(\exarg(\envv)\).
   The condition \(\gettags(\restrict{\Delta}{\FV(\form)}) \subseteq \set{\taguse}\)
   requires that all the free variables of \(\form\) are either integer variables or
   tagged with \(\taguse\), so that the extra argument \(\exarg(\envv)\) can be properly calculated.

   The rules from \rn{Tr-Var} to \rn{Tr-AppInt} just transform formulas in a compositional manner,
   with integer expressions unchanged.
   In \rn{Tr-Abs}, \(p_{x,\tagty}\) denotes the pattern defined by:
\[
  p_{x,\tagty} =
  \left\{\begin{array}{ll}
         (\exv{x},x)& \mbox{if $\tagty$ is of the form $(\efty,\taguse)$}\\
         x& \mbox{otherwise}\end{array}\right.\]
  For example, we have
  \[
  \judgesimp{\emptyset}{}{\lambda x.x\,1}{(\INT\to\Prop,\taguse)\to\Prop}{\lambda (v_x,x).x\,1}.
  \]
  Here, \(\emptyset\) denotes the empty type environment.
  In the rule \rn{Tr-Nu}, the auxiliary integer variable \(v_x\) associated with 
  \(x\) is prepared when \(\tagv=\taguse\). It is necessary in a case where \(x\) is passed
  to another function. For example,
  \(\nu x.f\,x\) (where \(f:(\INT\to\Prop,\taguse)\to\INT\to\Prop\)) is translated to
  \[\nu x.\letexp{v_x}{c v_f+d}{f\,(v_x,x)}.\]

  The key rule is \rn{Tr-Mu}.
  To see how \(\mu x.\form\) should be transformed,
  let us consider 
  \((\mu x.\form)\form_1\,\cdots\,\form_n\),
  where \(\mu x.\form\) is applied to actual
  arguments \(\form_1\,\cdots\,\form_n\). We estimate the number of unfoldings of
  \(\mu x.\form\) by gathering information from \((\mu x.\form)\form_1\,\cdots\,\form_n\).
  Thus, all the predicate variables in \(\mu x.\form\) and arguments should
  be tagged with \(\taguse\), as required by the third premise
  \(\gettags(\restrict{(\envv,y_1\COL\tagty'_1,\ldots,y_n\COL\tagty'_n)}{\FV(\form\,y_1\,\cdots\,y_n)})
  \subseteq \set{\taguse}\).
  To transform the subformula \(\form\), however, we need not require that the arguments of
  \(x\) should be tagged with \(\taguse\), when they are not passed to another
  least fixpoint formula in \(\form\).
  Thus, the types of arguments of \(\mu x.\form\) and those
  of \(x\) (inside \(\form\)) may be different
  in their outermost tags, as indicated in the second premise \(\tagty_i\raweq\tagty_i'\).
  For example, it is allowed that
  \(\mu x.\form\) has type \((\INT\to\Prop,\taguse)\to\Prop\) but
  \(x\) has type \((\INT\to\Prop,\tagnotuse)\to\Prop\) in \(\form\).
  The fourth premise (\(\form'''=\cdots\)) is analogous to the second premise of
  the rule \rn{Tr-Nu} explained above.
  The last premise (\(\form'=\cdots\)) takes care of the actual approximation of
  the \(\mu\)-formula by \(\nu\)-formula.
  The number of unfoldings of the \(\mu\)-formulas is represented by
\(\exarg(\restrict{\envv}{\FV(\mu x.\form)})\), and it is passed through
the extra parameter \(u\).

Finally, \rn{Tr-Sub} is the rule for subsumption, which allows, for example,
to convert a formula of type \((\efty,\tagnotuse)\to\Prop\) to that of type \((\efty,\taguse)\to\Prop\)
(but not in the opposite direction). The subtyping relation
\(\efty'\subtype\efty\leadsto\form\) is defined in \figref{fig:subtyping}.
As usual, the subtyping relation on predicate types is contravariant in the argument type,
and covariant in the return type. Since we need the corresponding coercion function
to achieve the transformation,
we have defined the subtyping relation as a ternary relation \(\efty'\subtype\efty\leadsto\form\),
where \(\form\) is a function to convert a formula of type \(\efty'\) to that of \(\efty\).

  

  %
\begin{figure}
    \infrule[Tr-TagT]{
      \judgesimp{\envv}{\rho}{\form}{\efty}{\form'}
      \andalso
      \gettags(\restrict{\Delta}{\FV(\form)}) \subseteq \set{\taguse}
    }{
      \judgesimp{\envv}{\rho}{\form}{(\efty,\taguse)}
                {(\exarg(\restrict{\envv}{\FV(\form)}),\form')}
    }
    \infrule[Tr-TagF]{
      \judgesimp{\envv}{\rho}{\form}{\efty}{\form'} 
    }{
      \judgesimp{\envv}{\rho}{\form}{(\efty,\tagnotuse)}{\form'}
    }
    \infax[Tr-Var]{
      \judgesimp{\fix{\envv}, \envpair{x}{\efty}{\tagv}}{\rho}{x}{\efty}{x}
    }
    \infrule[Tr-Or]{
      \judgesimp{\envv}{\rho}{\form_1}{\Prop}{\form_1'} \andalso
      \judgesimp{\envv}{\rho}{\form_2}{\Prop}{\form_2'}
    }{
      \judgesimp{\envv}{\rho}{\form_1\lor \form_2}{\Prop}{\form_1'\lor \form_2'}
    }
    %
    \infrule[Tr-And]{
      \judgesimp{\envv}{\rho}{\form_1}{\Prop}{\form_1'} \andalso
      \judgesimp{\envv}{\rho}{\form_2}{\Prop}{\form_2'}
    }{
      \judgesimp{\envv}{\rho}{\form_1\land \form_2}{\Prop}{\form_1'\land \form_2'}
    }
    %
    \infrule[Tr-Ge]{\ST(\envv)\pST e_1\ge e_2:\Prop
    }{
      \judgesimp{\envv}{\rho}{e_1\ge e_2}{\Prop}{e_1\ge e_2}
    }
    %
    \infrule[Tr-Abs]{
      \judgesimp{\envv,\,x\COL \tagty}{\rho}{\form}{\efty}{\form'}
    }{
      \judgesimp{\envv}{\rho}{\abs{x}{}{\form}}{\tagty\to\efty}{\abs{p_{x,\tagty}}{}{\form'}}
    }
    %
    \infrule[Tr-App]{
      \judgesimp{\envv}{\rho}{\form_1}{\tagty\to\efty}{\form_1'} \andalso
      \judgesimp{\envv}{\rho}{\form_2}{\tagty}{\form_2'}
    }{
      \judgesimp{\envv}{\rho}{\form_1\ \form_2}{\efty}{\form_1'\ \form_2'}
    }
    %
    \infrule[Tr-AppInt]{
      \judgesimp{\envv}{\rho}{\form}{\intfun{\efty}}{\form'}
      \andalso \ST(\envv)\pST e:\INT
    }{
      \judgesimp{\envv}{\rho}{\form\ e}{\efty}{\form'\ e}
    }
    \infrule[Tr-Nu]{
      \judgesimp{\envv,\envpair{x}{\efty}{\tagv}}{\rho}{\form}{\efty}{\form''}
      \\
      \form'=\left\{\begin{array}{ll}
      \letexp{\exv{x}}{\exarg(\restrict{\envv}{\FV(\nu x.\form)})}\form''
      &\mbox{if $\tagv=\taguse$}\\
      \form'' & \mbox{if $\tagv=\tagnotuse$}
      \end{array}\right.
    }{
      \judgesimp{\envv}{\rho}{\nu x.\form}{\efty}{\nu x.\form'}
    }
    %
    \infrule[Tr-Mu]{
      \judgesimp{\envv,\envpair{x}{\tagty_1\to\ldots\to\tagty_n\to\Prop}{\tagv}}{\rho}{\form}{\tagty_1\to\ldots\to\tagty_n\to\Prop}{\form''}\\
      \tagty_i\raweq\tagty_i'\mbox{ for each $i\in\set{1,\ldots,n}$}\\
      \gettags(\restrict{(\envv,y_1\COL\tagty'_1,\ldots,y_n\COL\tagty'_n)}{\FV(\form\,y_1\,\cdots\,y_n)})\subseteq\{\taguse\}\\
      \form'''=\left\{\begin{array}{ll}
      \letexp{\exv{x}}{\exarg(\restrict{\envv}{\FV(\mu x.\form)})}\form''
      &\mbox{if $\tagv=\taguse$}\\
      \form'' & \mbox{if $\tagv=\tagnotuse$}
      \end{array}\right.\\
      \form'=
      \big(\nu x.\lambda u.\lambda z_1\cdots z_n.u>0\land\hfill\\
      \qquad ([x(u-1)/x]\form''')\,z_1\,\cdots\,z_n\big)\,
      \exarg(\restrict{\envv}{\FV((\mu x.\form)y_1\,\cdots\,y_n)})
    }
            {
      \judgesimp{\envv}{\rho}{\mu x.\form}{
        \tagty_1'\to\ldots\to\tagty_n'\to\Prop}{\lambda p_{y_1,\tagty_1'}\cdots p_{y_n,\tagty_n'}.\form'\, p_{y_1,\tagty_1}\cdots\, p_{y_n,\tagty_n}}
    }
    %
    
    \infrule[Tr-Sub]{
      \judgesimp{\envv}{\rho}{\form}{\efty'}{\form'} \andalso
      \efty'\subtype\efty\leadsto \form_1
    }{
      \judgesimp{\envv}{\rho}{\form}{\efty}{\form_1(\form')}
    }
    \caption{Type-based Transformation Rules.}
    \label{fig:trans}
  \end{figure}
\begin{figure}
    \infax[S-Int]{
      \INT\subtype \INT \leadsto \lambda x.x
    }
    %
    \infax[S-Prop]{
      \Prop\subtype \Prop \leadsto \lambda x.x
    }
    %
    \infrule[S-Fun]{
      \tagty\subtype \tagty' \leadsto \form_1\andalso
      \efty'\subtype \efty \leadsto\form_2
    }{
      \tagty'\to\efty'\subtype \tagty\to\efty
      \leadsto \lambda f.\lambda \fix{p_{x,\tagty}}.\form_2(f(\form_1(\fix{p_{x,\tagty}})))
    }
    %
    \infrule[S-TaggedTT]{
      \efty'\subtype\efty\leadsto \form\\
    }{
      (\efty',T) \subtype (\efty,T)
      \leadsto \lambda (v_x,x).(v_x,\form\,x)
    }
    \infrule[S-TaggedTF]{
      \efty'\subtype\efty\leadsto \form\\
    }{
      (\efty',T) \subtype (\efty,F)
      \leadsto \lambda (v_x,x).\form\,x
    }
    \infrule[S-TaggedFF]{
      \efty'\subtype\efty\leadsto \form\\
    }{
      (\efty',F) \subtype (\efty,F)
      \leadsto \form
    }
    \caption{Subtyping Rules.}
    \label{fig:subtyping}
  \end{figure}
  
  \begin{example}
    \label{ex:tr}
    Recall the example of the formula \(\All\;(\lambda k.k\,0)\) in Section~\ref{sec:ho}.
    In the standard (non-equational) notation, the formula is expressed by
    \(\form_{\All}\;(\lambda k.k\,0)\) where:
    \[
    \begin{array}{l}
      \form_{\All}\defeq \nu \All.\lambda x.\form_F\,x \land \All\;(\form_{\Succ}\, x)\qquad
      \form_F \defeq \mu F.\lambda x.x(\lambda y.y=0) \lor F\;(\form_{\Pred}\,x)\\
      \form_{\Succ} \defeq \lambda x.\lambda k.x(\lambda y.k(y+1))\qquad
      \form_{\Pred} \defeq \lambda x.\lambda k.x(\lambda y.k(y-1)).
    \end{array}
    \]
    Let \(\efty_{\Succ}\) be
    \[\etyfun{\etyfun{\intfun{\Prop}}{\tagnotuse}{\Prop}}{\tagnotuse}{\etyfun{\intfun{\Prop}}{\tagnotuse}{\Prop}}.\]
    Then we have:
    \[
    \begin{array}{l}
      \judgesimp{\emptyset}{}{\form_{\Succ}}{\efty_{\Succ}}{\form_{\Succ}}\qquad
      \judgesimp{\emptyset}{}{\form_{\Pred}}{\efty_{\Succ}}{\form_{\Pred}}.
      \end{array}
    \]
    Let \(\efty_{F}\) and \(\efty'_{F}\) be defined by:
    \[
    \begin{array}{l}
     \efty_F\defeq \etyfun{\etyfun{\intfun{\Prop}}{\tagnotuse}{\Prop}}{\taguse}{\Prop}\qquad
     \efty'_F\defeq \etyfun{\etyfun{\intfun{\Prop}}{\tagnotuse}{\Prop}}{\tagnotuse}{\Prop}.
    \end{array}
    \]
    The body of \(\form_F\) is transformed to itself under \(F\COL(\efty'_F,\tagnotuse)\):
    \[
    F\COL(\efty'_F,\tagnotuse)\p
    \lambda x.x(\lambda y.y=0) \lor F\;(\form_{\Pred}\,x):\efty'_F\leadsto
    \lambda x.x(\lambda y.y=0) \lor F\;(\form_{\Pred}\,x).
    \]
    By using \rn{Tr-Mu}, we obtain
    \(
    \judgesimp{\emptyset}{}{\form_F}{\efty_F}{\form'_F}\),
    where
    \[\form'_F \defeq
    \lambda (\exv{x},x).
    (\nu F.\lambda u.\lambda x.u>0
    \land (\lambda x.x(\lambda y.y=0) \lor F\;(u-1)\;(\form_{\Pred}\,x))\;x)\,
    (c\exv{x}+d)\,x.\]
    Let \(\efty_{\All}\) be:
    \(\etyfun{\etyfun{\intfun{\Prop}}{\tagnotuse}{\Prop}}{\taguse}{\Prop}\).
    Then we have
    \(
    \judgesimp{\emptyset}{}{\form_{\All}}
              {\efty_{\All}}{\form'_{\All}}
     \)
     where
     \[
     \form'_{\All} \defeq \nu\All.\lambda (v_x,x).\form'_F\,(cv_x+d,x)\land\All\,(cv_x+d, \form_{\Succ}\,x).
     \]
     As a result, the whole formula \(\form_{\All}\;(\lambda k.k\,0)\) is translated to
     \( \form'_{\All}\;(d, \lambda k.k\,0)\).
     By rewriting the resulting formula in the equational form,
     we get \(\All\;(d,\lambda k.k\,0)\), where:
     \[
     \begin{array}{l}
       \All\,(v_x,x) =_\nu F'\,(cv_x+d,x)\land \All\,(cv_x+d,\Succ\,x)\qquad
       F'\,(v_x,x) =_\nu  F\,(c\exv{x}+d)\,x\\
       F\,u\,x=_\nu u>0
       \land (x(\lambda y.y=0) \lor F\;(u-1)\;(\Pred\,x))\qquad
       \cdots.
      \end{array}
     \]
     By inlining \(F'\), we can further simplify the equations to:
     \[
     \begin{array}{l}
       \All\,(v_x,x) =_\nu F\,(c^2\exv{x}+cd+d)\,x\land \All\,(cv_x+d,\Succ\,x)\\
       F\,u\,x=_\nu  u>0
       \land (x(\lambda y.y=0) \lor F\;(u-1)\;(\Pred\,x))\\
       \Succ\,x\,k=_\nu x(\lambda y.k(y+1))\qquad
       \Pred\,x\,k=_\nu x(\lambda y.k(y-1)).
     \end{array}
     \]
     \qed
     \end{example}

\iffull
\else
The following theorem states that our transformation provides a sound underapproximation
of \hflz{} formulas.
\begin{theorem}[soundness]
  \label{th:soundness}
  Suppose \(\judgesimp{\emptyset}{}{\form}{\Prop}{\form'}\).
  If \(\form'\) is valid, then so is \(\form\).
\end{theorem}

The transformation relation defined in the previous section was implicitly parameterized
by the constants \(c\) and \(d\).
To make them explicit, let us write \(\judgesimpc{\emptyset}{}{\form}{\Prop}{c,d}{\form'}\).
The theorem below states that the precision of the approximation is monotonic with
respect to \(c\) and \(d\).
The theorem justifies our approximation refinement process in \figref{fig:hflz},
which just increases the values of \(c\) and \(d\).
\begin{theorem}[monotonicity of the approximation]
  \label{th:monotonicity}
  Assume
  \(0\le c_1\le c_2\) and \(0\le d_1\le d_2\).
  Suppose 
  \(\judgesimpc{\emptyset}{}{\form}{\Prop}{c_1,d_1}{\form'^{(c_1,d_1)}}\) and
  \(\judgesimpc{\emptyset}{}{\form}{\Prop}{c_2,d_2}{\form'^{(c_2,d_2)}}\)
  are obtained by the same derivation except the values of \(c,d\).
    If \(\form'^{(c_1,d_1)}\) is valid, then so is \(\form'^{(c_2,d_2)}\).
\end{theorem}
Since the transformation rules are non-deterministic, we need to compare
\(\form'^{(c_1,d_1)}\) and \(\form'^{(c_2,d_2)}\) obtained by the same derivation
    in the theorem above.
    Because the shapes of possible derivations do not depend on the values of
    \(c\) and \(d\), we can keep using the same derivation during the approximation
    refinement process, to ensure that
    the approximation is always refined at each iteration of the refinement loop
    in \figref{fig:hflz}.
Proofs of the theorems above are found in
 a longer version of this paper~\cite{DBLP:journals/corr/abs-2203-07601}.
\fi

\begin{remark}
  Theorem~\ref{th:monotonicity} guarantees that the precision of the approximation monotonically
  increases, but does not guarantee that the approximation is precise enough for some \(c\) and \(d\).
  Indeed, there are cases where
\(\judgesimpc{\emptyset}{}{\form}{\Prop}{c,d}{\form'}\) and
\(\form\) is valid but \(\form'\) is invalid for any values of \(c\) and \(d\):
recall Remark~\ref{rem:ack}.
\end{remark}


\section{On the Power of Our Verification Method}
\label{sec:disc}

As discussed in Remark~\ref{rem:computability}, our
reduction from \hflz{} validity checking to \nuhflz{} validity
checking (hence also the overall verification method) is necessarily incomplete.
In this section, we 
compare our method (extended as sketched in Remark~\ref{rem:ack})
with previous automated
 methods for temporal property verification (especially termination and fair termination),
 and show that our method is
 strictly more powerful than
 previous methods based on (i) (lexicographic) linear ranking functions (LLRF)
 and those based on
 (ii) disjunctive well-founded relations with linear ranking functions (DWFLR).
 In other words, as mentioned in Remark~\ref{rem:complete-fragment},
 we characterize the class of \hflz{} formulas for which our method is (relatively)
 complete in terms of the classes of formulas for which previous methods are complete.
 As mentioned already, for first-order formulas (or programs),
 the idea of bounding the number of unfoldings (or recursive calls) to reduce
 termination/liveness properties to safety properties is not
 new~\cite{freqterm,DBLP:conf/sas/0001NIU19}, but the characterization
 of the power of such a method in terms of the popular
 methods using LLRF and DWFLR is new.
  
  The comparison with DWFLR is based on the observation that any sequence
 that conforms to DWFLR can be embedded into a monotonically decreasing sequence
 over \(\Nat^k\), which may be of independent interest.
 Below we consider only first-order formulas, as the issue of adding extra parameters
 (as discussed in Sections~\ref{sec:ho} and \ref{sec:opt}) is orthogonal to the discussion below.

\subsection{Methods Based on Well-Founded Relations Expressed as Linear Ranking Functions}
\label{sec:lrf}
Let us consider a formula \(X\), defined by
\(X\,y =_\mu \form(X)\) (i.e., \(X =\mu^{\INT\to\Prop}x.\lambda y.\form(x)\)),
where \(X\) does not occur in \(\form\), and
suppose that we wish to prove that \(X\,n\) holds for every integer \(n\).
Our approach was to approximate the formula \(X\,n\) with
\(\forall u\ge c |n|+d. X'\,u\,n\), where
\[X'\,u\,y=_\nu u>0\land \form(X'(u-1)).\]

An alternative approach (suggested, e.g., in \cite{DBLP:conf/lics/Nanjo0KT18,DBLP:conf/pepm/WatanabeTO019}
for fixpoint logics) is to pick a well-founded relation \(W\), and check that
the relation \(W\) holds between the arguments of recursive calls.
With this approach, the formula \(X\,n\) would be replaced with
\(X_{\LLRF}\,\infty\,n\), where
\[X_{\LLRF}\,y_p\,y =_\nu W(y,y_p)\land \form(\lambda y'.X_{\LLRF}\,y\,y'),\]
and \(\infty\) denotes a maximum integer with respect to the well-founded relation \(W\).
Here, the extra argument \(y_p\) has been added, which
represents the argument of the previous recursive call for \(X\); thus
it is checked that \(W\) holds between \(y\) and \(y_p\), and \(y_p\) has been
updated to \(y\) in the recursive use of \(X_{\LLRF}\) in \(\form\,(\lambda y'.X_{\LLRF}\,y\,y')\).

In \emph{automated} verification based on the latter approach, we have to fix a method
to pick an appropriate well-founded relation \(W\). The simplest approach is to select
a linear ranking function \(r(y)=c_r y+d_r\), let \(W(y,y_p)\) be
\(0\le r(y)<r(y_p)\),
and infer appropriate values for \(c_r\) and \(d_r\).
If \(X_{\LLRF}\,\infty\,n\) is valid, then the depth of recursion without violating
the relation \(W(y_p,y)\equiv 0\le r(y)<r(y_p)\) must be at most 
\(r(n)+1\; (=c_rn+d_r+1)\). Thus, \(X'\,(c|n|+d)\,n\) is also valid,
for \(c=|c_r|\) and \(d=|d_r|+1\). 
Thus, whenever the method based on linear ranking functions succeeds, our method should also succeed.

\begin{remark}
  \label{rem:lrf}
We have defined \(W(y,y_p)\) as \(0\le r(y)<r(y_p)\) above.
Alternatively, we could define \(W(y,y_p)\) as
  \(r(y)<r(y_p)\land 0\le r(y_p)\)~\cite{DBLP:conf/tacas/CookSZ13},
  so that the value of \(r(y)\) can be negative.
  We use the former definition for the sake of simplicity,
  but the latter definition can be obtained by setting
  \(r'(y) = \max(0, r(y)+1)\). This change does not affect the discussions
  below. In Sections~\ref{sec:lrf} and \ref{sec:llrf}, it suffices to
  increase the bound on the number of unfoldings in our approach by one,
  and in Section~\ref{sec:dwflrf}, it suffices to replace \(r(y)\) in the bound
  on the number of unfoldings with \(|r(y)|\).
  See also Example~\ref{ex:dwf}.
  \qed
\end{remark}

Furthermore, our method is superior to the linear ranking function approach,
in the following sense.
\begin{itemize}
  \item
There are formulas (or programs) for which our method succeeds but the approach of
linear ranking functions would fail.
Consider the following
recursive function \(f\) defined by:
\begin{align*}
  f\,y =& \IF\;y\le 0\;\THEN\; (\,)\; 
  \ELSE\; \IF\;y\;\MOD\;2=0\;\THEN\; f(y+1)\;\ELSE\;f(y-3).
\end{align*}
The termination of \(f(n)\) is represented by \(X\,n\), where:
\begin{align*}
  X\,y =_\mu & (y\le 0\imp\TRUE) \land\\
&  (y>0\imp ( (y\;\MOD\;2=0\imp X(y+1))\land (y\;\MOD\;2\ne 0\imp
  X(y-3)))).
\end{align*}
The formula \(X\,n\) is valid for all \(n\) (indeed, \(f\,n\) terminates for all \(n\)),
but since the argument of \(X\) goes up and down
(e.g. \(X(6)\to X(7)\to X(4)\to X(5)\to X(2)\to\cdots\)),
there exists no linear ranking function
\(r(y)\) such that \(X_{\LLRF}\,\infty\,n\) is valid.
In contrast, since the depth of required unfoldings of \(X\) (corresponding
to recursive calls for \(f\,n\)) is at most \(|n|+1\),
our approach succeeds
for any \(c\ge 1\) and \(d\ge 1\).
\item It is easier to systematically find appropriate values of \(c\) and \(d\),
  rather than to find the coefficients \(c_r\) and \(d_r\) for
  the ranking function. Recall that our approximation of an \hflz{} formula by a \nuhflz{} formula is
  \emph{monotonic} on \(c\) and \(d\) (Theorem~\ref{th:monotonicity}); thus, we just
  need to monotonically increase the values of \(c\) and \(d\), until the verification succeeds.
  In contrast, the precision of the ranking function approach is \emph{not} monotonic on the coefficients
  of ranking functions. For example, consider the termination
  of \(f\,y\,z\) where \(f\) is  defined by:
\begin{align*}
  f\,y\,z =_\mu & \IF\;y+z\le 0\;\THEN\; (\,)\;\ELSE\; f\,(y-2)\,(z+1).
\end{align*}
Then, the ranking function \(r(y,z)=y+z\) serves as a termination argument, but
the ranking function \(r'(y,z)=y+2z\), which has larger coefficients,
does NOT serve as a termination argument.
Thus, the search for appropriate ranking functions would require some heuristics.
\end{itemize}

\subsection{Methods Based on Lexicographic Linear Ranking Functions}
\label{sec:llrf}
The approach based on linear ranking functions discussed above is often too restrictive,
and a common approach for improvement is to use \emph{lexicographic linear ranking functions}~\cite{DBLP:conf/tacas/CookSZ13,DBLP:conf/cav/BradleyMS05}:
let \(r_1,\ldots,r_k\) be a sequence of linear ranking functions, and define
the well-found relation \(W\) by:
\[W(y,y_p) \IFF 0\le r_1(y)<r_1(y_p) \lor (0\le r_1(y)=r_1(y_p)\land 0\le r_2(y)<r_2(y_p))
\lor \cdots.
\]
For example, the termination argument for the Ackermann function
(given in Remark~\ref{rem:ack})
can be given by
\((r_1,r_2)\), where \(r_1(y,z) = y\) and \(r_2(y,z)=z\).

The extension discussed in Remark~\ref{rem:ack}
is at least as powerful as the method based on lexicographic linear ranking functions,
as discussed below. (We consider the case for \(k=2\) for the sake of simplicity;
the argument generalizes to an arbitrary sequence of
linear ranking functions \(r_1,\ldots,r_k\).)

Let us consider a formula \(X\)  defined by \(X\,\seq{y}=_\mu \form(X)\),
where \(X\) does not occur in \(\form\).
With the lexicographic linear ranking functions \(r_1,r_2\),
    \(X\,\seq{y}\) would
  be approximated by \(X_{\LLRF}\,\seq{y}_0\,\seq{y}\), where
  \(\seq{y}_0\) may be an arbitrary argument greater than \(\seq{y}\) with
  respect to \(W\), and
\[
X_{\LLRF}\,\seq{y}_p\,\seq{y} =_\nu
W(\seq{y}, \seq{y_p})\land
 \form\,(\lambda \seq{y}'.X_{\LLRF}\,\seq{y}\,\seq{y}').
\]
with \(W(\seq{y},\seq{y}_p)\equiv
0\le r_1(\seq{y})<r_1(\seq{y}_p) \lor (0\le r_1(\seq{y})=r_1(\seq{y}_p)\land
0\le r_2(\seq{y})<r_2(\seq{y}_p))\).

The extension discussed in Remark~\ref{rem:ack} (the special case of \(X_{\MC}\)
where \(k=2\)) instead approximates
\(X\,\seq{y}\) by
\(\forall u\ge c(|y_1|+\cdots+|y_\ell|)+d.
X'\,u\,u\,\seq{y}\)
where \(\seq{y}=y_1,\ldots,y_\ell\) and
\begin{align*}
X'\,u_1\,u_0\,\seq{y} =_\nu\; & u_1\ge 0\land u_0\ge 0\\
& \land
\form\,(\lambda \seq{y}'.
\forall u'_0\ge c(u_1+u_0+|y_1'|+\cdots+|y_\ell'|)+d.X'\,(u_1-1)\,u_0'\,\seq{y}')
\lor X'\,u_1\,(u_0-1)\,\seq{y}'.
\end{align*}

Suppose \(r_i(y_1,\ldots,y_\ell)=a_{i,1}y_1+\cdots+a_{i,\ell}y_\ell+b_i\) for \(i\in\set{1,2}\), and let \(c= \max_{i,j}(|a_{i,j}|)\) and \(d=\max(|b_1|,|b_2|)\),
so that \(r_i(y_1,\ldots,y_\ell)\le c(|y_1|+\cdots+|y_\ell|)+d\).
The following lemma ensures that our approximation (using \(X'\))
is at least as good as the method based on lexicographic linear ranking functions
(using \(X_{\LLRF}\)).
\begin{lemma}
  Let \(X'\) and \(X_{\LLRF}\) be the formulas as given above.
  For any integers \(\seq{n}_p\) and \(\seq{n}\), if
  \(m_1\ge r_1(\seq{n})\) and \(m_2\ge r_2(\seq{n})\), then
  \(\sem{X_{\LLRF}\,\seq{n}_p\,\seq{n}}\LEQ_\Prop
  \sem{X'\,m_1\,m_2\,\seq{n}}\).
\end{lemma}
\begin{proof}
  The proof proceeds by well-founded induction on \((r_1(\seq{n}),r_2(\seq{n}))\).
  Suppose \(\sem{X_{\LLRF}\,\seq{n}_p\,\seq{n}}=\Top\).
  Then it must be the case that \(W(\seq{n},\seq{n}_p)\), i.e.,
  \(0\le r_1(\seq{y})<r_1(\seq{y}_p) \lor (r_1(\seq{y})=r_1(\seq{y}_p)\land
0\le r_2(\seq{y})<r_2(\seq{y}_p))\).
  For every \(\seq{n}'\), if \(W(\seq{n},\seq{n}')\) does not hold,
  then \(\sem{X_{\LLRF}\,\seq{n}'\,\seq{n}}=\Bot\).
  Otherwise (i.e., if \(W(\seq{n},\seq{n}')\) holds),
  by the induction hypothesis,
  we have either
  \(\sem{X_{\LLRF}\,\seq{n}'\,\seq{n}}\LEQ_{\Prop} \sem{X'\,m_1\,(m_2-1)\,\seq{n}'}\)
   (if \(0\le r_1(\seq{n}')=r_1(\seq{n})\land 0\le r_2(\seq{n}')< r_2(\seq{n})\)), or
  \(\sem{X_{\LLRF}\,\seq{n}'\,\seq{n}}\LEQ_{\Prop} \sem{X'\,(m_1-1)\,m_2'\,\seq{n}'}\)
  (if \(0\le r_1(\seq{n}') < r_1(\seq{n})\))
  for any \(m_2'\ge r_2(\seq{n})\),
  which implies \[
  \sem{\lambda \seq{y}'.X_{\LLRF}\,\seq{n}\,\seq{y}'}
  \LEQ_{\seq{\INT}\to\Prop}
  \sem{\lambda \seq{y}'.
 X'\,m_1\,(m_2-1)\,\seq{y}'
 \lor \forall u'_2\ge c(|y_1'|+\cdots+|y_\ell'|)+d.X_{\LLRF}\,(m_1-1)\,u_2'\,\seq{y}'}.\]
 Thus, we have
 \begin{align*}
   \sem{X_{\LLRF}\,\seq{n}_p\,\seq{n}}
  & = \sem{W(\seq{n},\seq{n}_p)\land
   \form\,(\lambda \seq{y}'.X_{\LLRF}\,\seq{n}\,\seq{y}')}\\
 & 
   \LEQ_{\Prop}
   \sem{\form\,
(\lambda \seq{y}'.
 X'\,m_1\,(m_2-1)\,\seq{y}'
 \lor \forall u'_2\ge c(|y_1'|+\cdots+|y_\ell'|)+d.X_{\LLRF}\,(m_1-1)\,u_2'\,\seq{y}')}\\
 &= \sem{X'\,m_1\,m_2\,\seq{n}},
\end{align*}
as required.
\end{proof}

The argument above implies that our method (extended with
multiple counters as discussed above and in Remark~\ref{rem:ack})
is at least as powerful as the method based on lexicographic linear ranking functions.
Furthermore, the two points discussed at the end of Section~\ref{sec:lrf}
also apply to the comparison with lexicographic linear ranking functions.
Thus, our method is strictly more powerful,
and easier to automate, than 
the method based on lexicographic linear ranking functions.


\subsection{Methods Based on Disjunctive Well-Founded Relations with Linear Ranking Functions}
\label{sec:dwflrf}
  An alternative popular approach to proving termination or other liveness properties
  is to use disjunctive well-founded relations~\cite{DBLP:conf/lics/PodelskiR04,Kuwahara2014Termination}.
  In the context of the HFL model checking problem,
  the method can be recast as the following variation of \(X_{\LLRF}\) above:
  \[
X_{\DWF}\,\seq{y}_p\,\seq{y} =_\nu
D(\seq{y}, \seq{y_p})\land
 \form\,(\lambda \seq{y}'.X_{\DWF}\,\seq{y}\,\seq{y}'\land X_{\DWF}\;\seq{y}_p\;\seq{y}').
\]
Here, \(D\) is a finite union of well-founded relations.
The main difference from \(X_{\LLRF}\) is that the arguments of recursive calls
are compared between any ancestors and descendants, instead of just
between parents and children.
In practice, a linear ranking function is often used to represent
each well-founded relation composing \(D\).
It is known that the method based on disjunctive well-founded relations (DWF)
with linear ranking functions (called DWFLR below)
is more powerful than the method based on disjunctive well-founded relations
(in fact, the example discussed at the end of Section~\ref{sec:lrf} can be
handled by the former)~\cite{DBLP:conf/tacas/CookSZ13}, but
Cook et al.~\cite{DBLP:conf/tacas/CookSZ13} have empirically shown that
the latter is often more efficient than the former.

There exists an example for which our method works, but
DWFLR does not. Consider the predicate \(X\) defined by:
\[
X\,y\,a\,p =_\mu y=a \lor X\,((y\times a)\;\mathtt{mod}\;p)\;a\;p,
\]
and suppose that we wish to prove that for any positive
integer \(a>0\) and prime number \(p\), \(X\,(a^2)\,a\,p\) holds.
Since \(a^p\equiv a\; \mathtt{mod}\; p\) (Fermat's little theorem),
it suffices to approximate
\(X\,(a^2)\,a\,p\) with \(X'\,(p+1)\,(a^2)\,a\,p\) where 
\[X'\,u\,y\,a\,p =_\nu u>0\land
((y=a)\lor X'\,(u-1)\,((y\times a)\;\mathtt{mod}\;p)\;a\;p)\]
in our method.
However, there exists no appropriate disjunctive well-founded relation
that can be expressed as a combination of linear ranking functions.

Below we show that our method is strictly more powerful than DWFLR.
Suppose that \(D= R_{r_1}\cup\cdots \cup R_{r_k}\), where
\(r_i\) is an affine function (i.e., \(r_i(y_1,\ldots,y_\ell)\) is
of the form \(c_0+c_1y_1+\cdots +c_\ell y_\ell\))
and 
\(R_{r_i} = \set{(\seq{v}, \seq{w}) \mid 0 \le r_i(\seq{v}) < r_i(\seq{w})}\).
We show that any sequence \(\set{\seq{v}_i}_{0\le i\le m-1}\) such that
\(\forall i,j. i<j \imp (v_i,v_j)\in D\) can be mapped to
a decreasing sequence over \(\Nat^k\) (where \(\Nat\) is the set of
natural numbers) with respect to the lexicographic order on \(\Nat^k\).

We first prepare some definitions.
Given a set \(V \subseteq \Nat^\ell\), we write
\(\Down{V}\subseteq \Nat^\ell\) for the set:
\[
\set{\seq{v}\in\Nat^\ell \mid \forall \seq{w}\in V.(\seq{v},\seq{w})\in D}.
\]
Note that \(\seq{v}_{j+1} \in \Down{V_j}\) holds
for any sequence
\(\set{\seq{v}_i}_{1\le i\le m}\) that satisfies the above condition,
where \(V_j := \set{\seq{v}_i\mid 1\le i\le j}\).
For a tuple \(\seq{b} = (b_1,\ldots,b_k)\) with \(b_i\in \Nat\cup\set{\omega}\),
we define the \emph{base set} \(\Base{(b_1,\ldots,b_k)}\) by:
\[
\Base{(b_1,\ldots,b_k)} :=
\set{\set{v}\in\Nat^\ell\mid
    \forall i\in\set{1,\ldots,k}.b_i=\omega \lor 0\le r_i(\set{v})= b_i}.
\]

For \(j\in\set{0,\ldots,m-1}\), 
we can construct a finite set \(S_j\subseteq (\Nat\cup\set{\omega})^k\)
such that \(\Down{V_j}\subseteq \bigcup_{\seq{b}\in S_j} \Base{\seq{b}}\)
(in which case we say \(S_j\) \emph{covers} \(\Down{V_j}\)).
Given \(\seq{v}_0 \in \Nat^\ell\), we set \(S_0\) to
\begin{align*}
  \bigcup_{i\in\set{1,\ldots,k}}
   \set{(\omega^{i-1},x,\omega^{k-i}) \mid 0\le x < r_i(\seq{v}_0)}.
\end{align*}
Suppose that \(S_{j-1}\) covers \(\Down{V_{j-1}}\) (with \(j\ge 1\)), i.e.,
  \(\Down{V_{j-1}}\subseteq \bigcup_{\seq{b}\in S_{j-1}} \Base{\seq{b}}\),
  and \(\seq{v}_j \in \Down{V_{j-1}}\).
  Since
  \(\seq{v}_j \in
  \Down{V_{j-1}}\subseteq \bigcup_{\seq{b}\in S_{j-1}} \Base{\seq{b}}\),
  we can pick \(\seq{b}'=(b'_1,\ldots,b'_k)\) such that
  \(\seq{v}_j \in\Base{\seq{b}'}\),
  and let \(I_{\seq{b}'} = \set{i\in\set{1,\ldots,k}\mid b'_i=\omega}\).
  If \(I_{\seq{b}'}=\emptyset\), then 
  let \(S_j\) be \(S_{j-1}\setminus \set{\seq{b}'}\).
  Otherwise, let 
  \begin{align*}
  S_j :=  ( S_{j-1}\setminus \set{\seq{b}'})
    \cup
    \bigcup_{i\in I_{\seq{b}'}} \set{(b'_1,\ldots,b'_{i-1},x,b'_{i+1},\ldots,b'_k)
      \mid 0\le x < r_i(\seq{v}_j)}.
  \end{align*}
  Then \(\Down{V_j}\) is covered by \(S_j\), as required. A concrete example
  of \(S_j\) is given in Example~\ref{ex:cover} below.

  Now, let us define the measure \(\measure{S_j}\) of \(S_j\) as
  \((\meas{S_j}{k-1},\ldots,\meas{S_j}{0})\) where
  \(\meas{S}{i} = |\set{\seq{b}\in S\mid |I_{\seq{b}}|=i}|\).
  Intuitively, \(\meas{S}{i}\) denotes the number of \(i\)-dimensional hyperplanes
  used to cover \(\Down{V_i}\).
  By the construction of \(S_j\) above, \(\measure{S_0},\measure{S_1},\ldots\)
  forms a monotonically decreasing sequence with respect to
  the lexicographic ordering on \(\Nat^k\). Furthermore, whenever an \(i\)-dimensional
  hyperplane is removed, the number of (\(i-1\))-dimensional
  hyperplanes added to the covering is bounded above by \(\Sigma_{\ell\in \set{1,\ldots,k}} r_\ell(\seq{v})\).

  By the observation above, \(X'\,(r_1(\seq{v})+\cdots + r_k(\seq{v}), 0^{k-1})\,\seq{v}\)
  is at least as good an approximation of \(X\,\seq{v}\) as
  \(X_{\DWF}\,\seq{\infty}\,\seq{v}\), where \(X'\) is defined by:
  \begin{align*}
    X'\,(u_{k-1},\ldots,u_0)\,\seq{y}
    =_\nu\, & u_{k-1}\ge 0\land\cdots \land u_0\ge 0\land\\&
    \form(\lambda \seq{y}'.
    X'\,(u_{k-1}-1,u_{k-2}+r_1(\seq{y}')+\cdots + r_k(\seq{y}'), \ldots,u_0)\,\seq{y}'\\&\quad
    \lor 
    X'\,(u_{k-1},u_{k-2}-1, u_{k-3}+r_1(\seq{y}')+\cdots + r_k(\seq{y}'), \ldots,u_0)\,\seq{y}'\\&\quad
    \lor 
    \cdots
    \lor 
    X'\,(u_{k-1},\ldots,u_1, u_0-1)_k\,\seq{y}' 
    ).
  \end{align*}
  Thus, our approximation (using \(X_{\MC}\) in Remark~\ref{rem:ack})
  is at least as good as DWFLR if
  we set \(c\ge 1\) and \(d\) so that \(c(|y_1|+\cdots+|y_\ell|)+d \ge
  r_1(y_1,\ldots,y_\ell)+\cdots + r_k(y_1,\ldots,y_\ell)\).

  \begin{example}
    \label{ex:cover}
    Let \(k=\ell=3\), \(\seq{v}_i\,(0\le i\le 4)\)
    and \(r_j(x,y,z)\, (j\in\set{1,2,3})\) be given as follows.
    \begin{align*}
    & \seq{v}_0=(1,1,1),\quad
    \seq{v}_1=(0,2,1),\quad
    \seq{v}_2=(2,0,1),\quad
    \seq{v}_3=(2,2,0),\quad
    \seq{v}_4=(1,0,1),\\
    & r_1(x,y,z)=x,\quad r_2(x,y,z)=y,\quad r_3(x,y,z)=z.
    \end{align*}
    Then, \(S_j\)  and \(\measure{S_j}\) (\(0\le i\le 4\)) are:
    \allowdisplaybreaks[4]
    \begin{align*}
      &  S_0 = \set{(0,\omega,\omega), (\omega,0,\omega), (\omega,\omega,0)}      & \measure{S_0}=(3,0,0)\;\qquad\ \  \\&
      S_1 = \set{(0,1,\omega),(0,0,\omega),(0,\omega,0), (\omega,0,\omega), (\omega,\omega,0)}
      & \measure{S_1}=(2,3,0)\;\qquad\ \  \\&
      S_2 = \set{(0,1,\omega),(0,0,\omega),(0,\omega,0), (1,0,\omega), (\omega,0,0), (\omega,\omega,0)} & \measure{S_2}=(1,5,0)\;\qquad\ \  \\&
      S_3 = \set{(0,1,\omega),(0,0,\omega),(0,\omega,0), (1,0,\omega), (\omega,0,0), (1,\omega,0),(\omega,1,0)} & \measure{S_3}=(0,7,0)\;\qquad\ \  \\&
      S_4 = \set{(0,1,\omega),(0,0,\omega),(0,\omega,0), (1,0,0), (\omega,0,0), (1,\omega,0),(\omega,1,0)} & \measure{S_4}=(0,6,1). 
    \end{align*}
   \end{example}

  \begin{example}
    \label{ex:dwf}
    Let us consider the termination of the following
    imperative program~\cite{DBLP:conf/tacas/CookSZ13}.
\begin{verbatim}
      assume(m>0); while x<>m do if x>m then x:=0 else x:=x+1.
\end{verbatim}
\newcommand\Loop{\texttt{Loop}}
The termination is expressed as the validity of the \hflz{} formula
\(\forall m. \forall x.m>0 \imp \texttt{Loop}\;x\;m\) where
\begin{align*}
  \Loop\;x\;m =_\mu x=m \lor (x>m\land \Loop\;0\;m)\lor
  (x\le m\land \Loop\;(x+1)\;m).
\end{align*}
Cook et al.~\cite{DBLP:conf/tacas/CookSZ13} gave the following disjunctive
well-founded relation as the termination argument:
\[
(x < x_p \land 0\le x_p) \lor (m-x < m_p-x_p \land 0\le m_p-x_p).
\]
Let \(r_1(x,m) = \max(0,x+1)\) and \(r_2(x,m)=\max(0,m-x+1)\)
(recall Remark~\ref{rem:lrf}).
Based on the discussion above, the formula can be approximated in
our approach by:
\(\forall m. \forall x.m>0 \imp \Loop'\;(r_1(x,m)+r_2(x,m),0)\;x\;m\) where:
\begin{align*}
&  \Loop'\;(u_1,u_0)\;x\; m =_\nu
  u_1\ge 0\land u_0\ge 0\\& \qquad\land
  \form(\lambda x'.\lambda m'.\Loop'(u_1-1,u_0+r_1(x',m')+r_2(x',m'))\;x'\;m'
  \lor \Loop'(u_1,u_0-1)\;x'\;m'),
\end{align*}
with \(\form \equiv \lambda p.x=m \lor (x>m\land p\;0\;m)\lor
(x\le m\land p\;(x+1)\;m)\).
The part \(r_1(x,m)+r_2(x,m)\) can be replaced by
\(|x+1|+|m-x+1| \le (|x|+1)+(|m|+|x|+1) = 2|x|+|m|+2\).
This approximation is a conservative one obtained from the theory above;
the approximation by
\(\forall m. \forall x.m>0 \imp \Loop''\;(|x|+|m|+2)\;x\;m\) where:
\begin{align*}
  \Loop''\;u\;x\;m =_\nu
  u\ge 0\land
  \form(\Loop''\;(u-1))
\end{align*}
would actually suffice.
\qed
\end{example}


\section{Implementation and Experiments}
\label{sec:exp}

\subsection{Implementation}
\label{sec:imp}
We have implemented an \hflz{} validity checker \muhfl{}
based on the method described in Section~\ref{sec:method}
(including the extension discoursed in Remark~\ref{rem:ack}).
We use \rethfl{}
~\cite{DBLP:conf/aplas/KatsuraIKT20}
as the backend \nuhflz{} solver.\footnote{Our solver can also use other \nuhflz{} solvers, such as \pahfl~\cite{DBLP:conf/sas/IwayamaKST20},
  but \rethfl{} has performed the best in our use case.}
For the sake of simplicity of the implementation,
our current implementation does not support the subsumption rule \rn{Tr-Sub}
in the type-based transformation described in Section~\ref{sec:opt};
the lack of TR-SUB may miss some optimization opportunity in theory,
but we have not observed any problem caused by it
in the experiments reported in Section~\ref{sec:eval}.

In the implementation, the coefficients for bounds for the number of
unfoldings (\(c\), \(d\)), the coefficients for extra arguments
(\(c'\), \(d'\)), and the number of counters described in Remark
\ref{rem:ack} are set as shown in Table
\ref{tab:PrametersForEachIteration} for the first four iterations of
the approximation. After the fourth iteration, \(c\), \(d\), \(c'\), and \(d'\) are
doubled for every two iterations, and the number of counters
alternates between 1 and 2.

\begin{table}[tbp]
  \caption{Parameters for each iteration of the approximation.}
  \label{tab:PrametersForEachIteration}
  \begin{center}
    \begin{tabular}{|c|c|c|c|c|c|}
    \hline
    iteration & \(c\) & \(d\) & \(c'\) & \(d'\) & the number of counters \\ \hline
    1         & 1 & 2  & 1  & 1  & 1                                      \\
    2         & 1 & 2  & 1  & 1  & 2                                      \\
    3         & 1 & 16 & 1  & 1  & 1                                      \\
    4         & 1 & 16 & 1  & 1  & 2                                      \\ \hline
    \end{tabular}
  \end{center}
\end{table}

\begin{remark}
  As discussed already, 
  the precision of the approximation monotonically increases
  with respect to \(c,\,d,\,c',d'\) and the number of counters.
  In practice, however, choosing large values for
  \(c,\,d,\,c',d'\) and the number of counters may slow down the backend \nuhflz{} solver.
  Thus, it would be better to choose different values for those parameters
  for each least fixpoint formula.
  Developing a better way to determine the values is left for future work. \qed
\end{remark}

We additionally implemented an optimization to omit some extra arguments for consecutive higher-order arguments.
For example, consider a formula \(\lambda x.\lambda y.\form\) whose type is \((\efty, {\taguse})\to(\efty, {\taguse})\to\Prop\).
If partial applications of this formula never occur, we can transform the formula to \(\lambda (v_{x,y},x,y).\form\), instead of \(\lambda (v_{x},x).\lambda (v_{y},y).\form\), because the two arguments are always passed together.
We infer which extra arguments can be omitted by using a type-based analysis.
In this section, we call this optimization ``Optimization 2,'' and the optimization described in Section~\ref{sec:opt} ``Optimization 1''.

\nk{I have removed the comment about disjunction elimination.}

\subsection{Evaluation}
\label{sec:eval}
To evaluate the effectiveness of our method,
we conducted the following \fix{three} experiments:
\begin{itemize}
  \item comparison with previous verification tools for temporal properties of higher-order programs~\cite{Kuwahara2014Termination,Kuwahara2015Nonterm,MTSUK16POPL,Watanabe16ICFP},
  \item comparison with and without the two optimizations of extra arguments, and
  \item  \fix{further evaluation of our tool using \hflz{} formulas that are reduced from temporal property verification problems for higher-order programs~\cite{Koskinen14,Hofmann14CSL,Lester2011,DBLP:conf/pepm/WatanabeTO019}, which cannot be solved (at least directly)
    by the previous verification methods used in the first experiment.}
    This benchmark set includes
    branching-time properties of higher-order programs, for which there were no automated
    tools to our knowledge.
\end{itemize}
Note that all the problems used in the experiments involve higher-order predicates;
thus, the previous tool for
the first-order fragment of \hflz{}~\cite{DBLP:conf/sas/0001NIU19} is not applicable.

The experiments were conducted on a machine
with Intel Xeon CPU E5-2680 v3 and 64GB of RAM.
We set the timeout to 900 seconds.
The benchmark instances
are available at \url{https://github.com/hopv/hflz-benchmark} and the docker image containing the source code and the binary of \muhfl{} is available at \url{https://www.kb.is.s.u-tokyo.ac.jp/vm-images/popl-2023-muapprox.tar.gz}.

\subsubsection{Comparison with previous higher-order program verification tools}
\label{sec:expComparisonWithOtherTools}

We compared
\muhfl{} with the previous automated verification tools for temporal properties
of higher-order programs~\cite{Kuwahara2014Termination,Kuwahara2015Nonterm,MTSUK16POPL,Watanabe16ICFP}.
We used the following six benchmark sets consisting of verification problems for
OCaml programs.
\begin{itemize}
\item \texttt{termination}: termination verification problems taken from
  ~\cite{Kuwahara2014Termination}.
\item \texttt{non-termination}: non-termination verification problems
  taken from~\cite{Kuwahara2015Nonterm}.
\item \texttt{fair-termination}: verification problems for fair termination,
  taken from ~\cite{MTSUK16POPL}.
\item \texttt{fair-non-termination}: verification problems for fair
  non-termination, taken from~\cite{Watanabe16ICFP}.
\item \texttt{termination-ho}: a variation of \texttt{termination},
  where integer values have been converted to closures
  in a manner similar to the example in Section~\ref{sec:ho}.
\item \texttt{fair-termination-ho}: a variation of \texttt{fair-termination},
  where integer values have been converted to closures.
\end{itemize}
%
From the benchmark sets, we excluded out instances that can be directly
reduced to \nuhflz{} formulas (i.e., formulas without the least fixpoint operator \(\mu\)).

For each verification problem instance in the benchmark sets, we
 (automatically) converted it to the \hflz{} validity checking problem
by using the reductions in \cite{ESOP2018full,DBLP:conf/pepm/WatanabeTO019},
and ran our tool \muhfl{}.
We compared its performance with the result of running 
the corresponding previous verification tool
(e.g. Kuwahara et al.'s tool~\cite{Kuwahara2014Termination} for
\texttt{termination} and \texttt{termination-ho}).

\begin{figure}[tp]
  \begin{center}
    \includegraphics[scale=0.8]{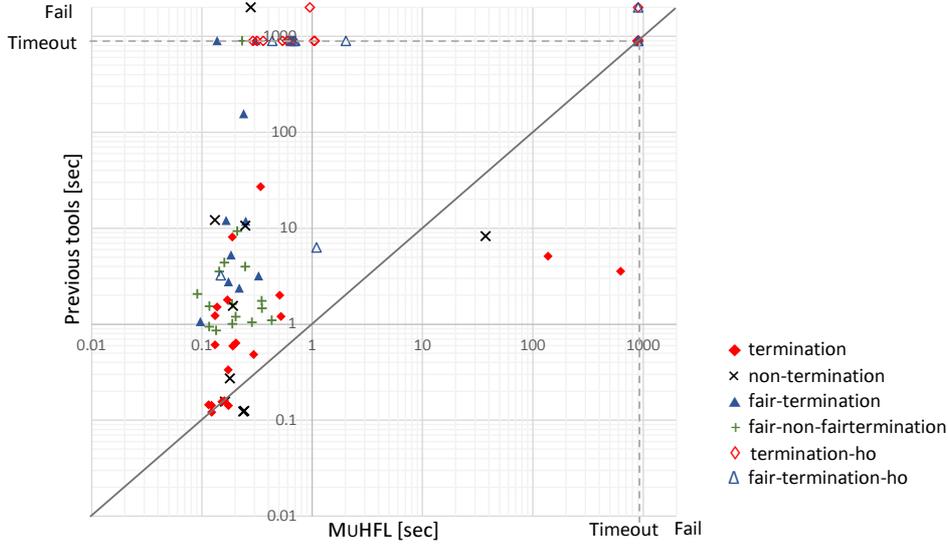}
  \end{center}
  \caption{Comparison with previous tools.}
  \label{tab:ComparisonWithExistingTools}
\end{figure}

\fix{
\begin{table}[tbp]
  \caption{The number of solved instances per benchmark set.}
  \label{tab:NumberOfSolvedInstancesPerBenchmarkSet}
  \begin{center}
\begin{tabular}{|l|l|l|l|} \hline
benchmark set        & no. of instances & solved by \muhfl{} & solved by previous tools \\ \hline
\texttt{termination}          & 21               & 20              & 20                       \\
\texttt{non-termination}      & 9                & 9               & 8                        \\
\texttt{fair-termination}     & 10               & 10              & 8                        \\
\texttt{fair-non-termination} & 16               & 16              & 15                       \\
\texttt{termination-ho}       & 21               & 11              & 0                        \\
\texttt{fair-termination-ho}  & 10               & 7               & 2                        \\ \hline
total                & 87               & 73              & 53                       \\  \hline
\end{tabular}
\end{center}
\end{table}
}

The result is summarized in \figref{tab:ComparisonWithExistingTools}.
In the figure, ``Fail'' means that
the tool was aborted with some error.
The number of solved instances per benchmark set is shown in Table~\ref{tab:NumberOfSolvedInstancesPerBenchmarkSet}.
In total, our tool \muhfl{} could solve more problems than (the combination of)
the four previous tools.
In particular, for the benchmark sets with more higher-order values
 (\texttt{termination-ho} and \texttt{fair-termination-ho}),
the previous tools could solve only two instances
while \muhfl{} solved \fix{18} instances.
We believe that the failure of \muhfl{} to solve the \fix{14} instances is mainly due to
the current limitations of the backend solver \rethfl{}, rather than a fundamental
limitation of our approach.
In fact, we confirmed that most of the \fix{14} instances  
could be solved
by the backend solver after some manual \fix{transformation} of the \nuhflz{} formulas
generated by our reduction from \hflz{} to \nuhflz{}.

As for the instances for which both the previous tools and \muhfl{} succeed,
\fix{\muhfl{} were often faster than the previous tools.}
There are three outliers in \figref{tab:ComparisonWithExistingTools},
for which our tool is significantly slower:
one non-termination problem and two termination problems.
One of those instances requires a large value for the constant \(d\),
hence requiring the number of iterations of the approximation
(cf. Table~\ref{tab:PrametersForEachIteration}).
For the other two, the generated \nuhflz{} formulas belong to a class of formulas
which the current backend solver \rethfl{} is not good at (specifically,
the class of formulas that contain disjunction on fixpoint formulas).
This problem can be remedied by a further improvement of the backend solver.


Overall, our tool \muhfl{} outperformed the previous tools,
which is remarkable, considering that the previous tools 
were specialized for
particular verification problems (such as termination and non-termination), whereas
our tool can deal with all of those verification problems in a uniform manner.
More details of the experimental results are given in
\iffull
\else
a longer version~\cite{DBLP:journals/corr/abs-2203-07601}.
\fi

\subsubsection{Comparison with and without the optimizations of extra arguments}

To evaluate the effectiveness of the optimizations on extra arguments (Optimization 1 discussed in
Section~\ref{sec:opt} and Optimization 2 explained in Section~\ref{sec:imp}),
we compared the running times of our tool with and without the optimizations.
From the previous experiments, we picked
instances for which extra arguments are required and \muhfl{} successfully terminated.

The comparison of the total running times (measured in seconds)
is shown in \fix{Fig.}~\ref{fig:ComparisonWithAndWithoutOptimization}.
Fig.~\ref{fig:ComparisonWithAndWithoutOptimizationBackend} shows the times taken by the backend solver,
and \fix{Fig.}~\ref{fig:ComparisonExtraWithAndWithoutOptimization}
compares the numbers of extra parameters.
The instances \fix{from \texttt{murase-closure-ho} to \texttt{koskinen-4-ho}} are from \texttt{fair-termination-ho},
and the other instances are from \texttt{termination-ho}.
As observed in
\fix{Fig.}~\ref{fig:ComparisonWithAndWithoutOptimization},
 the optimizations were generally effective;
\fix{the two optimizations reduced the total running times except for \texttt{sum-ho}, with a maximum reduction of 2.1 seconds for \texttt{binomial-ho}.}
For \fix{all instances}, the number of extra parameters was
reduced \fix{to less than half by the optimizations,} which we think is
one of the reasons for the reduction of the running time.
\fix{The optimizations were particularly effective for \texttt{binomial-ho}.
The reason could be that the number of extra parameters for it was significantly reduced by the optimizations.}
The number of iterations of the approximation refinement cycle
required to solve the instances \fix{was one for all instances and was} not changed by the optimizations.
The optimizations (in particular, Optimization 2 alone for \fix{\texttt{fibonacci-ho}})
sometimes increased the running times by up to a factor of 12.
That seems to be due to some unexpected behavior of the backend solver.%
\footnote{For example, sometimes just a renaming of predicates
  significantly changes the running time of a backend CHC solver used inside \rethfl{}.}
%
\fix{
\begin{figure}[tbp]
  \begin{center}
    \includegraphics[scale=0.8]{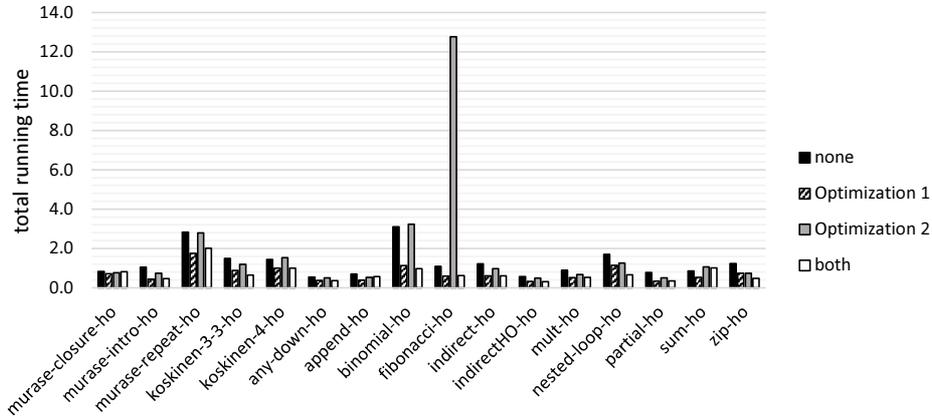}
  \end{center}
  \caption{Comparison of total running times with and without the optimizations.}
  \label{fig:ComparisonWithAndWithoutOptimization}
\end{figure}
\begin{figure}[tbp]
  \begin{center}
    \includegraphics[scale=0.8]{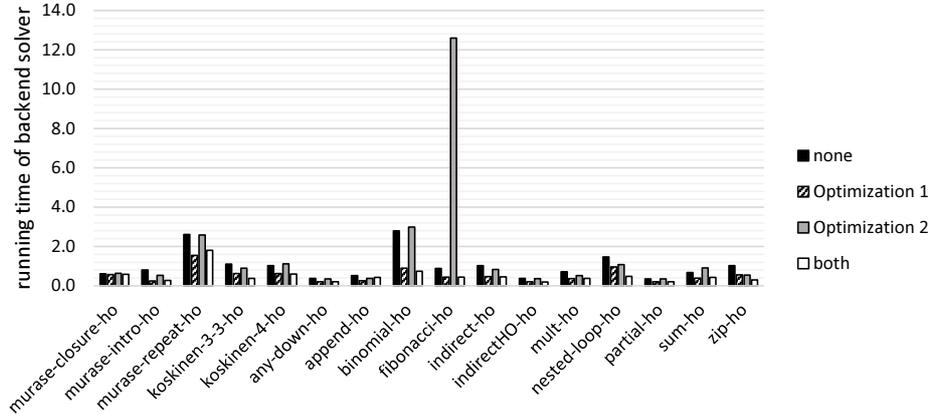}
  \end{center}
  \caption{Comparison of running times of the backend solver with and without the optimizations.}
  \label{fig:ComparisonWithAndWithoutOptimizationBackend}
\end{figure}
\begin{figure}[tbp]
  \begin{center}
    \includegraphics[scale=0.8]{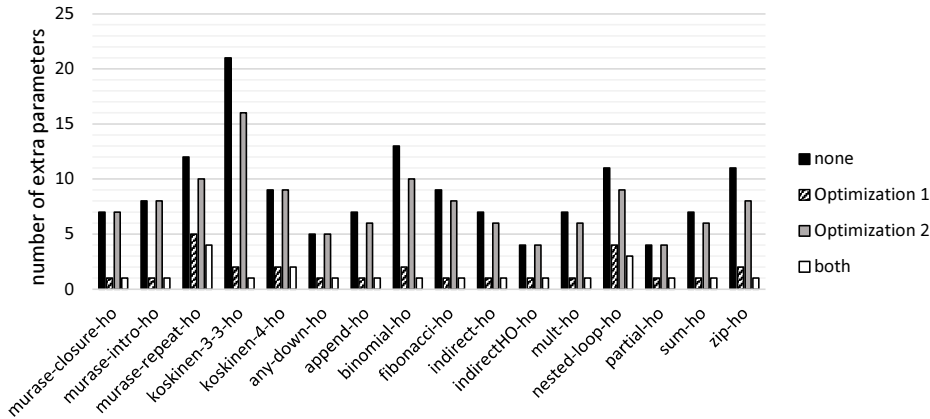}
  \end{center}
  \caption{Comparison of the number of extra parameters with and without the optimizations.}
  \label{fig:ComparisonExtraWithAndWithoutOptimization}
\end{figure}}
\fix{
  \subsubsection{Solving
    other temporal verification problems for higher-order programs}

  To test our tool for other temporal verification problems,
  we have collected a new benchmark set from previous papers
  on temporal verification of higher-order programs~\cite{Hofmann14CSL,Koskinen14,Lester2011,DBLP:conf/pepm/WatanabeTO019}.
  We automatically converted the problems to the \hflz{} validity checking problem by
  using the reduction of Watanabe et al.~\cite{DBLP:conf/pepm/WatanabeTO019},
and ran \muhfl{}.
The results are shown in Table \ref{tab:exp4}, where times are shown in seconds.
The instance \texttt{hofmann-2-direct} comes from the work of
Hofmann and Chen~\cite{Hofmann14CSL}. Specifically, it is the second example
from Appendix A of the extended technical report~\cite{Hofmann14CSLCoRR}.
The instances \texttt{koskinen-1-direct}, \texttt{koskinen-2-direct}, \texttt{koskinen-3-direct}, and \texttt{koskinen-4-direct} are from Figure 10 of the paper by Koskinen and Terauchi~\cite{Koskinen14},
where \texttt{koskinen-1-direct} is \textsc{Reduce}, \texttt{koskinen-2-direct} is \textsc{Rumble},
\texttt{koskinen-3-direct} are \textsc{Eventually Global}, and \texttt{koskinen-4-direct} is \textsc{Alternate Inevitability}.
The instances whose names contain \texttt{buggy} are variations of the Koskinen and Terauchi's instances,
where the original instances are modified so that
the specified properties are violated.
The instance \texttt{lester-direct} is from Appendix H.1 of the paper
by Lester et al.~\cite{Lester2011}.
The previous verification tools used in the first experiment cannot directly solve
those instances~\footnote{Some of the instances were
   \emph{manually} translated to fair-(non-)termination problems
   and were used in the experiments reported in \cite{MTSUK16POPL,%
     Watanabe16ICFP}.}.
The last three instances, i.e., \texttt{sas19-tab1-24-ho}, \texttt{repeat}, and \texttt{repeat2},
are about branching-time properties of higher-order programs,
for which there were no previous automated tools to our knowledge.
The instance \texttt{sas19-tab1-24-ho} is a higher-order version of
the corresponding instance used in \cite{DBLP:conf/sas/0001NIU19} (\#24 in Table~1).
The instance \texttt{repeat} has been taken from \cite{DBLP:conf/pepm/WatanabeTO019}
(Example~3.3) and \texttt{repeat2} is a variation of it.
Watanabe et al.~\cite{DBLP:conf/pepm/WatanabeTO019} proved the validity of
\texttt{repeat} manually by using Coq, but our tool can now prove it fully
automatically.

As shown in Table \ref{tab:exp4}, all the instances were successfully solved by our
tool.
Note that, among the previous work~\cite{Hofmann14CSL,Koskinen14,Lester2011,DBLP:conf/pepm/WatanabeTO019} which
those instances come from, only Lester et al.~\cite{Lester2011} implemented an actual
automated verification tool,
which can solve \texttt{lester-direct} (but not others)
in 0.035 second. For the other instances,
we are not aware of other fully automated tools that can directly solve them.

%
\begin{table}[tbp]
  \caption{Results and running times for solving \hflz{} formulas translated from verification problems for higher-order programs.}
  \label{tab:exp4}
  \begin{center}
\begin{tabular}{|l|l|l|} \hline
instance                   & result  & running time \\ \hline
\texttt{hofmann-2-direct}          & valid   & 0.27         \\
\texttt{koskinen-1-buggy-direct}   & invalid & 0.31         \\
\texttt{koskinen-1-direct}         & valid   & 0.19         \\
\texttt{koskinen-2-buggy-direct}   & invalid & 1.38         \\
\texttt{koskinen-2-direct}         & valid   & 0.70         \\
\texttt{koskinen-3-buggy-1-direct} & invalid & 0.32         \\
\texttt{koskinen-3-buggy-2-direct} & invalid & 0.39         \\
\texttt{koskinen-3-direct}         & valid   & 0.37         \\
\texttt{koskinen-4-direct}         & valid   & 0.38         \\
\texttt{lester-direct}             & valid   & 0.22         \\
\texttt{sas19-tab1-24-ho}                     & valid   & 0.21         \\
\texttt{repeat}                    & valid   & 6.99         \\
\texttt{repeat2}                   & valid   & 8.05         \\  \hline
\end{tabular}
\end{center}
\end{table}
}

\section{Related Work}
\label{sec:related}

As already mentioned,
the framework of program verification by reduction to \hflz{} validity checking
has been advocated by Kobayashi et al.~\cite{KTW18ESOP,DBLP:conf/pepm/WatanabeTO019}.
It can be considered a generalization of the CHC-based program
verification framework~\cite{Bjorner15}, where higher-order predicates and
fixpoint alternations are allowed in the target logic.
Burn et al.~\cite{DBLP:journals/pacmpl/BurnOR18} also considered a higher-order extension of CHC
and its application to program verification, but their logic corresponds to
the \nuhflz{} fragment, which does not support fixpoint alternations.
Higher-order predicates
are useful for modeling higher-order programs, and fixpoint alternations are
useful for dealing with temporal properties. The effectiveness of the
framework has been partially demonstrated for the first-order
fragment of \hflz{}~\cite{DBLP:conf/sas/0001NIU19}.
There have been implementations of automated validity checkers for
\nuhflz{} (a fragment of \hflz{} without least fixpoint operators)~\cite{DBLP:conf/sas/IwayamaKST20,DBLP:conf/aplas/KatsuraIKT20},
and a satisfiability checker for HoCHC~\cite{DBLP:journals/pacmpl/BurnOR18}.
To our knowledge, however, there have been no tools for full \hflz{}.

Watanabe et al. (\cite{DBLP:conf/pepm/WatanabeTO019}, Section~4.2) sketched (but have not
implemented) another method for approximating least fixpoint formulas with
greatest fixpoint formulas. Their method relies on the discovery of a well-founded
relation on the arguments of fixpoint predicates, which is hard to automate, especially
in the presence of higher-order arguments.
In contrast, our approach is much easier to automate; to refine the approximation,
we just need to monotonically increase constant parameters (\(c,d,c',d'\)
in Section~\ref{sec:method}).
The idea of our approach has been inspired
by the work of Kobayashi et al.~\cite{DBLP:conf/sas/0001NIU19} on
the first-order fragment of \hflz{}.
That idea can further be traced back to
the method of Fedyukovich et al.~\cite{freqterm} for termination analysis.

The idea of adding extra integer parameters for higher-order arguments
has been inspired by Unno et al.'s work~\cite{UnnoTK13} on a relatively
complete refinement type system, but the details on the way extra parameters
are different. In particular, our method of adding extra parameters is
easier to automate. We have also proposed a type-based optimization to avoid
redundant extra parameters. Our type-based optimization may be considered
an instance of type-based flow analysis~\cite{Palsberg01,NielsonBook}.

Various techniques have been proposed and implemented for automated verification of
various linear-time temporal properties of higher-order programs, including safety
properties~\cite{Jhala08,Terauchi10POPL,KSU11PLDI,SUK13PEPM,Ong11POPL,DBLP:journals/pacmpl/PavlinovicSW21,zhu_2015},
termination~\cite{Kuwahara2014Termination}, non-termination~\cite{Kuwahara2015Nonterm},
fair termination~\cite{MTSUK16POPL}, and fair non-termination~\cite{Watanabe16ICFP}.
In contrast to those studies, which developed separate techniques
and tools for
proving different properties, our \hflz{} validity checker serves as a common backend for
all of those properties, and can also be used for the verification of
branching-time properties of higher-order programs.

\section{Conclusion}
\label{sec:conc}

We have proposed an automated method for \hflz{} validity checking,
which provides a streamlined approach to fully automated verification
of temporal properties of higher-order programs, and proved the soundness
of our method. We have also compared our approach with previous verification
methods for proving termination and liveness properties, such as those using
lexicographic linear ranking functions and disjunctively well-founded relations.
We have implemented a tool based on the proposed method, and
confirmed its effectiveness through experiments. To our knowledge,
our tool is the first automated \hflz{} validity checker,
which serves as a common backend tool for automated verification of temporal
properties of functional programs.

\subsection*{Acknowledgment}
This work was supported by
JSPS KAKENHI Grant Numbers JP20H05703.


\newpage
\section*{Appendix}
\appendix
\section{Correctness of the Transformation}
\label{sec:correctness}

  We show correctness of the transformation given in Section~\ref{sec:opt}.
  
  We first show that the output of the transformation is a well-typed formula
  (which implies, in particular, extra variables are appropriately passed around).
  Since we have extended the syntax of the target language with pairs,
  we extend simple types by:
  \[
\begin{array}{l}
\sty \mbox{ (extended simple types)} ::= \INT \mid \fty \mid \INT\times \fty\\
\fty \mbox{ (extended predicate types)} ::= \Prop \mid \sty \to \fty.
\end{array}
\]
and
 extend the typing rules in \figref{fig:st} with the following rules.
\infrule[T-Pair]{\stenv\pST e:\INT\andalso \stenv\pSTex \form:\fty}
      {\stenv\pSTex (e,\form):\INT\times\fty}
\infrule[T-PAbs]{\stenv, v_x:\INT, x:\fty_1\pSTex \form:\fty_2}
        {\stenv\pSTex \lambda (v_x,x).\form: \INT\times\fty_1\to\fty_2}

For tagged types \(\tagty\) and \(\efty\), the corresponding simple types
\(\trT{\tagty}\) and \(\trT{\efty}\) are defined by:
\[
\begin{array}{l}
\trT{(\efty,\taguse)} = \INT\times \trT{\efty}\\
\trT{(\efty,\tagnotuse)} = \trT{\efty}\\
\trT{\INT} = \INT\\
\trT{(\tagty\to\efty)}=\trT{\tagty}\to\trT{\efty}.
\end{array}
\]
We extend the operation to type environments by:
\[\trT{(x_1\COL\tagty_1,\ldots,x_k\COL\tagty_k)} =
p_{x_1,\tagty_1}\COL\trT{\tagty_1},\ldots,p_{x_k,\tagty_k}\COL\trT{\tagty_k}.\]
Here, \((v_x,x)\COL\INT\times\fty\) is considered a shorthand for \(v_x\COL\INT,x\COL\fty\).
For example, \[
\trT{(x\COL (\Prop,\tagnotuse), y\COL(\INT\to\Prop,\taguse))}
= x\COL\Prop, v_y\COL\INT, y\COL\INT\to\Prop.\]

The following lemma states that the output of the transformation is a well-typed
formula.
\begin{lemma}
  If \(\judgesimp{\envv}{}{\form}{\fty}{\form'}\),
  then \(\trT{\envv}\pST \form':\trT{\fty}\).
  If \(\judgesimp{\envv}{}{\form}{\tagty}{\form'}\),
  then \(\trT{\envv}\pST \form':\trT{\tagty}\).
  In particular,
  \(\judgesimp{\emptyset}{}{\form}{\Prop}{\form'}\)
  implies \(\emptyset\pST \form':\Prop\).
\end{lemma}
\begin{proof}
  This follows by straightforward induction on the derivations of
  \(\judgesimp{\envv}{}{\form}{\fty}{\form'}\) and
  \(\judgesimp{\envv}{}{\form}{\tagty}{\form'}\).
\end{proof}


To prove Theorem~\ref{th:soundness},
we extend the semantics of \hflz{} formulas defined in Section \ref{sec:pre} with pairs introduced in Section~\ref{sec:ho}.
\[
\begin{array}{l}
  \semd{\INT\times\fty} = \set{(n,w)\mid n\in \Z,w\in\semd{\fty}}\\
  \LEQ_{\INT\times\fty} = \set{((n,w),(n,w'))\in \semd{\INT\times\fty}\times
    \semd{\INT\times\fty}
    \mid 
    w\LEQ_{\fty}w'}\\
  \sem{\Gamma\pST (e,\form):\INT\times\fty}(\rho)=
  (\sem{\Gamma\pST e:\INT}\rho,  \sem{\Gamma\pST \form:\fty}\rho)\\
  \sem{\Gamma\pST \lambda (v_x,x^{\fty_1}).\form: \INT\times\fty_1\to\fty}{\rho}
  = \\\qquad
  \lambda (n,w)\in \Z\times \semd{\fty_1}.
  \sem{\Gamma,v_x\COL\INT, x\COL\fty_1\pST\form:\fty}(\rho\set{v_x\mapsto n, x\mapsto w})\\
\end{array}  
\]
We define
the approximation relation
\(\Simge_{\efty}\subseteq \semd{\ST(\efty)}\times \semd{\trT{\efty}}\) by:
\[
\begin{array}{l}
  \Simge_{\INT} = \set{(n,n)\mid n\in \Z}\\
  \Simge_{\Prop} = \set{(\Top,\Bot),(\Top,\Top),(\Bot,\Bot)}\\
  \Simge_{(\efty,\tagnotuse)} = \Simge_{\efty}\\
  \Simge_{(\efty,\taguse)} = \set{(w, (n,w'))\mid n\in \Z, w\Simge_{\efty}w'}\\
  \Simge_{\tagty\to\efty} =
   \set{(f,f')\mid \forall w,w'.w\Simge_{\tagty}w'\imp f\,w\Simge_{\efty}f'\,w'}.
\end{array}  
\]
For \(\envv\), we define \(\Simge_{\envv}\subseteq \semd{\ST(\envv)}\times \semd{\trT{\envv}}\) by:
\[\rho \Simge_{\envv}\rho'
\IFF \forall x\in\dom(\envv).\rho(x)\Simge_{\envv(x)}\rho'(p_{x,\envv(x)}).\]
\begin{lemma}
  \label{lem:soundness-of-trans}
  If \(\judgesimp{\envv}{}{\form}{\efty}{\form'}\)
  and \(\rho \Simge_{\envv} \rho'\), then
  \[\sem{\ST(\envv)\pST \form:\ST(\efty)}\rho
     \Simge_{\efty} \sem{\trT{\envv}\pST \form':\trT{\efty}}\rho'.\]
\end{lemma}
\begin{proof}
  This follows by induction on the derivation of
  \(\judgesimp{\envv}{}{\form}{\efty}{\form'}\), with case analysis on the last
  rule. Since the other cases are trivial, we discuss only the case for \rn{Tr-Mu}.
  Suppose that the last rule used for deriving
  \(\judgesimp{\envv}{}{\form}{\efty}{\form'}\) 
  is \rn{Tr-Mu}.
  Then, we have:
  \[
  \begin{array}{l}
    \form=\mu x.\form_1\\
    \efty=\tagty_1\to\cdots\to\tagty_n\to\Prop\\
    \efty'=\tagty'_1\to\cdots\to\tagty'_n\to\Prop\\
    \judgesimp{\envv,\envpair{x}{\efty'}{\tagv}}{\rho}{\form_1}{\efty'}{\form_1''}\\
      \tagty_i\raweq\tagty_i'\mbox{ for each $i\in\set{1,\ldots,n}$}\\
      \gettags(\restrict{(\envv,y_1\COL\tagty_1,\ldots,y_n\COL\tagty_n)}{\FV(\form_1\,y_1\,\cdots\,y_n)})\subseteq \{\taguse\}\\
      \form_1'''=\left\{\begin{array}{ll}
      \letexp{\exv{x}}{\exarg(\restrict{\envv}{\FV(\mu x.\form_1)})}\form_1''
      &\mbox{if $\tagv=\taguse$}\\
      \form_1'' & \mbox{if $\tagv=\tagnotuse$}
      \end{array}\right.\\
      \form_1'=
      \big(\nu x.\lambda u.\lambda z_1\cdots z_n.u>0\land\hfill\\
      \qquad ([x(u-1)/x]\form_1''')\,z_1\,\cdots\,z_n\big)\,
      \exarg(\restrict{\envv}{\FV((\mu x.\form)y_1\,\cdots\,y_n)})\\
    \form'=\lambda p_{y_1,\tagty_1}\cdots p_{y_n,\tagty_n}.\form_1'\, p_{y_1,\tagty_1'}\cdots\, p_{y_n,\tagty_n'}\\
   \end{array}
  \]
  By the induction hypothesis, for any \(w,w'\) such that
  \(w\Simge_{(\efty',t)} w'\), we have:
  \[
  \begin{array}{l}\sem{\ST(\envv),x\COL\ST(\efty')\pST \form_1:\ST(\efty')}\rho\set{x\mapsto w}\\
  \Simge_{\efty} \sem{\trT{\envv},\trT{(x\COL(\efty',t))}
    \pST \form_1'':\trT{\efty'}}\rho'\set{p_{x,(\efty',t)}\mapsto w'}.
  \end{array}\]
  Let \(w''\) be the second element of \(w'\) if \(t=\taguse\), and \(w''=w'\) otherwise.
  Then, from the relation above and the definition of \(\form_1'''\), we obtain:
  \[
  \begin{array}{l}\sem{\ST(\envv),x\COL\ST(\efty')\pST \form_1:\ST(\efty')}\rho\set{x\mapsto w}\\
  \Simge_{\efty} \sem{\trT{\envv}, x\COL \trT{\efty'}
    \pST \form_1''':\trT{\efty'}}\rho'\set{x\mapsto w''}.
  \end{array}
  \]
  Therefore, we have:
  \[
  \begin{array}{l}
    \lambda w\in\semd{\ST(\efty')}.
    \sem{\ST(\envv),x\COL\ST(\efty')\pST \form_1:\ST(\efty')}\rho\set{x\mapsto w}\\
    \Simge_{(\efty',\tagnotuse)\to\efty'}
    \lambda w''\in\semd{\trT{\efty'}}.
    \sem{\trT{\envv}, x\COL \trT{\efty'}
    \pST \form_1''':\trT{\efty'}}\rho'\set{x\mapsto w''}.
  \end{array}
  \]
  Let \(m\) be
  \(\sem{\trT{\envv}\pST\exarg(\restrict{\envv}{\FV((\mu x.\form)y_1\,\cdots\,y_n)}):\INT}\rho'\).
  It follows by easy induction on \(m\) that
  \[
  \sem{\trT{\envv}\pST\form_1':\trT{\efty'}}\rho'
  =
  (\lambda w\in\semd{\trT{\efty'}}.
    \sem{\trT{\envv}, x\COL \trT{\efty'}
    \pST \form_1''':\trT{\efty'}}\rho'\set{x\mapsto w})^m (\Bot_{\trT{\efty'}}).
    \]
    Thus, we have:
    \[
    \begin{array}{l}
      \sem{\ST(\envv)\pST \form:\ST(\efty)}\rho\\
      \GEQ_{\ST(\efty)}
(\lambda w\in\semd{\ST(\efty')}.
      \sem{\ST(\envv),x\COL\ST(\efty')\pST \form_1:\ST(\efty')}\rho\set{x\mapsto w})^m (\Bot_{\ST(\efty')})\\
      \Simge_{\efty'}
  (\lambda w\in\semd{\trT{\efty'}}.
    \sem{\trT{\envv}, x\COL \trT{\efty'}
      \pST \form_1''':\trT{\efty'}}\rho'\set{x\mapsto w})^m (\Bot_{\trT{\efty'}})\\
= \sem{\trT{\envv}\pST\form_1':\trT{\efty'}}\rho'.    
    \end{array}
    \]
    Therefore, we have
      \(\sem{\ST(\envv)\pST \form:\ST(\efty)}\rho
     \Simge_{\efty} \sem{\trT{\envv}\pST \form':\trT{\efty}}\rho'\) as required.
\end{proof}

Theorem~\ref{th:soundness} follows as an immediate corollary of the above lemma.
\begin{proof}[Proof of Theorem~\ref{th:soundness}]
  A special case of Lemma~\ref{lem:soundness-of-trans},
  where \(\envv=\emptyset\) and \(\efty=\Prop\).
\end{proof}


To prove Theorem~\ref{th:monotonicity}, we define another family
of relations \(\set{\LEmono_{\efty}}_\efty\) parameterized by \(\efty\).
\[
\begin{array}{l}
  \LEmono_{\INT} = \set{(n,n)\mid n\in \Z}\\
  \LEmono_{\Prop} = \set{(\Bot,\Top),(\Top,\Top),(\Bot,\Bot)}\\
  \LEmono_{(\efty,\tagnotuse)} = \LEmono_{\efty}\\
  \LEmono_{(\efty,\taguse)} = \set{((n,w), (n',w'))\mid n\le n', w\LEmono_{\efty}w'}\\
  \LEmono_{\tagty\to\efty} =
   \set{(f,f')\mid \forall w,w'.w\LEmono_{\tagty}w'\imp f\,w\LEmono_{\efty}f'\,w'}
\end{array}
\]
We write \(\rho\LEmono_{\envv}\rho'\) if \(\rho(p_{x,\envv(x)})\LEmono_{\envv(x)}\rho'(p_{x,\envv(x)})\) for
every \(x\in\dom(\envv)\).
\begin{lemma}
  \label{lem:monotonicity}
  Suppose \(\judgesimpc{\envv}{}{\form}{\efty}{c_1,d_1}{\form'^{(c_1,d_1)}}\) and
  \(\judgesimpc{\envv}{}{\form}{\efty}{c_2,d_2}{\form'^{(c_2,d_2)}}\)
  are derived from the same derivation except the values of \(c,d\).
  Suppose also \(\rho\LEmono_{\envv}\rho'\).
  If \(0\le c_1\le c_2\) and \(0\le d_1\le d_2\), then
  \[\sem{\trT{\envv}\pST \form'^{(c_1,d_1)}:\trT{\efty}}\rho
     \LEmono_{\efty}
\sem{\trT{\envv}\pST \form'^{(c_2,d_2)}:\trT{\efty}}\rho'.\]
\end{lemma}
\begin{proof}
  This follows by induction on the derivation of
  \(\judgesimpc{\envv}{}{\form}{\efty}{c_1,d_1}{\form'^{(c_1,d_1)}}\), with case analysis on
  the last rule.
  We discuss only the case for \rn{Tr-Mu}, since the other cases are trivial.
  In the case for \rn{Tr-Mu}, we have:
  \[
  \begin{array}{l}
    \form=\mu x.\form_1\\
    \efty=\tagty_1\to\cdots\to\tagty_n\to\Prop\\
    \efty'=\tagty'_1\to\cdots\to\tagty'_n\to\Prop\\
    \judgesimpc{\envv,\envpair{x}{\efty'}{\tagv}}{\rho}{\form_1}{\efty'}{c,d}{\form_1''^{(c,d)}}\\
      \tagty_i\raweq\tagty_i'\mbox{ for each $i\in\set{1,\ldots,n}$}\\
      \gettags(\restrict{(\envv,y_1\COL\tagty_1,\ldots,y_n\COL\tagty_n)}{\FV(\form_1\,y_1\,\cdots\,y_n)})=\{\taguse\}\\
      \form_1'''^{(c,d)}=\left\{\begin{array}{ll}
      \letexp{\exv{x}}{\exarg^{(c,d)}(\restrict{\envv}{\FV(\mu x.\form_1)})}\form_1''^{(c,d)}
      &\mbox{if $\tagv=\taguse$}\\
      \form_1''^{(c,d)} & \mbox{if $\tagv=\tagnotuse$}
      \end{array}\right.\\
      \form_1'^{(c,d)}=
      \big(\nu x.\lambda u.\lambda z_1\cdots z_n.u>0\land\hfill\\
      \qquad ([x(u-1)/x]\form_1'''^{(c,d)})\,z_1\,\cdots\,z_n\big)\,
      \exarg^{(c,d)}(\restrict{\envv}{\FV((\mu x.\form_1)y_1\,\cdots\,y_n)})\\
    \form'^{(c,d)}=\lambda p_{y_1,\tagty_1}\cdots p_{y_n,\tagty_n}.\form_1'^{(c,d)}\, p_{y_1,\tagty_1'}\cdots\, p_{y_n,\tagty_n'}\\
   \end{array}
  \]
  for \((c,d)\in\set{(c_1,d_1),(c_2,d_2)}\).
  Here, we have made \(c,d\) explicit in \(\exarg\).
  Suppose \(\rho\LEmono_{\envv}\rho'\).
  By the induction hypothesis,
  for any \(w\LEmono_{(\efty',t)}w'\),
  we have
  \[
\begin{array}{l}
  \sem{\trT{\envv},p_{x,(\efty',t)}\COL\trT{(\efty',t)}\pST \form_1''^{(c_1,d_1)}:\trT{\efty'}}
  (\rho\set{p_{x,(\efty',t)}\mapsto w})\\
    \LEmono_{\efty'}
    \sem{\trT{\envv},p_{x,(\efty',t)}\COL\trT{(\efty',t)}\pST \form_1''^{(c_2,d_2)}:\trT{\efty'}}
    (\rho'\set{p_{x,(\efty',t)}\mapsto w'}).
\end{array}
\]
Let \(w_1,w_1'\) be the second component of \(w,w'\) if \(\tagv=\taguse\)
and \(w_1=w, w_1'=w'\) otherwise.
Since
\[
\begin{array}{l}
\sem{\trT(\envv)\pST\exarg^{(c_1,d_1)}(\restrict{\envv}{\FV(\mu x.\form_1)}):\INT}\rho\\
\le
\sem{\trT(\envv)\pST\exarg^{(c_2,d_2)}(\restrict{\envv}{\FV(\mu x.\form_1)}):\INT}\rho',
\end{array}\]
we have:
  \[
\begin{array}{l}
  \sem{\trT{\envv},x\COL\trT{\efty'}\pST \form_1'''^{(c_1,d_1)}:\trT{\efty'}}
  (\rho\set{x\mapsto w_1})\\
    \LEmono_{\efty'}
    \sem{\trT{\envv},x\COL\trT{\efty'}\pST \form_1'''^{(c_2,d_2)}:\trT{\efty'}}
    (\rho'\set{x\mapsto w_1'}).
\end{array}
\]
Let \(m\) and \(m'\) be
\(\sem{\trT{\envv}\pST
  \exarg^{(c_1,d_1)}(\restrict{\envv}{\FV((\mu x.\form_1)y_1\,\cdots\,y_n)}):\INT}\rho\)
and 
\(\sem{\trT{\envv}\pST
  \exarg^{(c_2,d_2)}(\allowbreak\restrict{\envv}{\FV((\mu x.\form_1)y_1\,\cdots\,y_n)}):\INT}\rho'\)
respectively. Since \(m\le m'\), we have:
\[
\begin{array}{l}
  \sem{\trT{\envv}\pST\form_1'^{(c_1,d_1)}:\trT{\efty'}}\rho\\
  =(\lambda w_1.\sem{\trT{\envv},x\COL\trT{\efty'}\pST \form_1'''^{(c_1,d_1)}:\trT{\efty'}}
  (\rho\set{x\mapsto w_1}))^m(\Bot_{\trT{\efty'}})\\
  \LEmono_{\efty'}
(\lambda w_1.\sem{\trT{\envv},x\COL\trT{\efty'}\pST \form_1'''^{(c_1,d_1)}:\trT{\efty'}}
  (\rho\set{x\mapsto w_1}))^{m'}(\Bot_{\trT{\efty'}})\\  
  \LEmono_{\efty'}
(\lambda w_1'.\sem{\trT{\envv},x\COL\trT{\efty'}\pST \form_1'''^{(c_2,d_2)}:\trT{\efty'}}
  (\rho'\set{x\mapsto w_1'}))^{m'}(\Bot_{\trT{\efty'}})\\
  =
  \sem{\trT{\envv}\pST\form_1'^{(c_2,d_2)}:\trT{\efty'}}\rho'.
\end{array}
\]
We have thus
  \[\sem{\trT{\envv}\pST \form'^{(c_1,d_1)}:\trT{\efty}}\rho
     \LEmono_{\efty}
     \sem{\trT{\envv}\pST \form'^{(c_2,d_2)}:\trT{\efty}}\rho'\]
     as required.
\end{proof}

Theorem~\ref{th:monotonicity} is an immediate corollary of the above lemma.
\begin{proof}[Proof of Theorem~\ref{th:monotonicity}]
  A special case of Lemma~\ref{lem:monotonicity},
  where \(\envv=\emptyset\) and \(\efty=\Prop\).
\end{proof}

\pagenumbering{arabic}
\section{More Information on the Experimental Results}
\label{app:expdetails}
The full results of the experiment for comparison with previous higher-order program verification tools are shown in Table \ref{tab:resultextendedone} and \ref{tab:resultextendedtwo}.
For all instances, the expected result is ``valid.''

{\small
\setlength{\tabcolsep}{1mm}
\begin{table}[bp]
\caption{Results of the experiment for comparison with previous higher-order program verification tools (1/2).}
\label{tab:resultextendedone}
\begin{center}
\begin{tabular}{|l|l|l|r|l|l|l|r|} \hline
                     &                      & \multicolumn{4}{l|}{\muhfl{}} & \multicolumn{2}{l|}{previous   tools} \\ \hline
benchmark            & instance             & result  & \multicolumn{1}{l|}{time} &
  \begin{tabular}[c]{@{}l@{}}no. of iter.\\ (prover)\end{tabular} & \begin{tabular}[c]{@{}l@{}}no. of iter.\\ (disprover)\end{tabular} & result & \multicolumn{1}{l|}{time} \\ \hline
\texttt{termination}          & \texttt{ackermann}            & valid   & 136.52                   & 2      & 1    & verified  & 5.13      \\
\texttt{termination}          & \texttt{any-down}             & valid   & 0.17                     & 1      & 1    & verified  & 0.14      \\
\texttt{termination}          & \texttt{append}               & valid   & 0.12                     & 1      & 1    & verified  & 0.14      \\
\texttt{termination}          & \texttt{binomial}             & valid   & 0.20                     & 1      & 1    & verified  & 0.64      \\
\texttt{termination}          & \texttt{fibonacci}            & valid   & 0.16                     & 1      & 1    & verified  & 0.16      \\
\texttt{termination}          & \texttt{foldr}                & valid   & 0.13                     & 1      & 1    & verified  & 1.23      \\
\texttt{termination}          & \texttt{indirect}             & valid   & 0.52                     & 1      & 1    & verified  & 1.20      \\
\texttt{termination}          & \texttt{indirectHO}           & valid   & 0.19                     & 1      & 1    & verified  & 8.12      \\
\texttt{termination}          & \texttt{indirectIntro}        & valid   & 0.34                     & 1      & 1    & verified  & 27.03     \\
\texttt{termination}          & \texttt{loop2}                & valid   & 0.29                     & 2      & 1    & verified  & 0.48      \\
\texttt{termination}          & \texttt{map}                  & valid   & 0.51                     & 1      & 1    & verified  & 2.00      \\
\texttt{termination}          & \texttt{mc91}                 & valid   & 623.44                   & 5      & 3    & verified  & 3.57      \\
\texttt{termination}          & \texttt{mult}                 & valid   & 0.11                     & 1      & 1    & verified  & 0.15      \\
\texttt{termination}          & \texttt{nested-loop}          & valid   & 0.17                     & 1      & 1    & verified  & 0.33      \\
\texttt{termination}          & \texttt{partial}              & valid   & 0.14                     & 1      & 1    & verified  & 1.51      \\
\texttt{termination}          & \texttt{quicksort}            & timeout & -                        & 5      & 3    & timeout   & -         \\
\texttt{termination}          & \texttt{sum}                  & valid   & 0.12                     & 1      & 1    & verified  & 0.12      \\
\texttt{termination}          & \texttt{toChurch}             & valid   & 0.13                     & 1      & 1    & verified  & 0.61      \\
\texttt{termination}          & \texttt{up-down}              & valid   & 0.19                     & 1      & 1    & verified  & 0.59      \\
\texttt{termination}          & \texttt{x-plus-2-n}           & valid   & 0.17                     & 1      & 1    & verified  & 1.79      \\
\texttt{termination}          & \texttt{zip}                  & valid   & 0.15                     & 1      & 1    & verified  & 0.16      \\
\texttt{non-termination}      & \texttt{fib-CPS-nonterm}      & valid   & 0.16                     & 1      & 1    & verified  & 0.16      \\
\texttt{non-termination}      & \texttt{fixpoint-nonterm}     & valid   & 0.18                     & 1      & 2    & verified  & 0.27      \\
\texttt{non-termination}      & \texttt{foldr-nonterm}        & valid   & 0.28                     & 1      & 2    & fail      & -         \\
\texttt{non-termination}      & \texttt{indirectHO-e}         & valid   & 0.24                     & 1      & 2    & verified  & 0.13      \\
\texttt{non-termination}      & \texttt{indirect-e}           & valid   & 0.24                     & 1      & 2    & verified  & 0.12      \\
\texttt{non-termination}      & \texttt{inf-closure}          & valid   & 0.25                     & 1      & 1    & verified  & 10.66     \\
\texttt{non-termination}      & \texttt{loopHO}               & valid   & 0.19                     & 1      & 2    & verified  & 1.55      \\
\texttt{non-termination}      & \texttt{passing-cond}         & valid   & 37.15                    & 1      & 2    & verified  & 8.30      \\
\texttt{non-termination}      & \texttt{unfoldr-nonterm}      & valid   & 0.13                     & 1      & 1    & verified  & 12.15     \\
\texttt{fair-termination}     & \texttt{murase-closure}       & valid   & 0.16                     & 1      & 1    & verified  & 12.05     \\
\texttt{fair-termination}     & \texttt{murase-intro}         & valid   & 0.25                     & 1      & 1    & verified  & 11.87     \\
\texttt{fair-termination}     & \texttt{murase-repeat}        & valid   & 0.22                     & 1      & 1    & verified  & 2.37      \\
\texttt{fair-termination}     & \texttt{hofmann-2}            & valid   & 0.10                     & 1      & 1    & verified  & 1.07      \\
\texttt{fair-termination}     & \texttt{koskinen-1}           & valid   & 0.14                     & 1      & 1    & timeout   & -         \\
\texttt{fair-termination}     & \texttt{koskinen-2}           & valid   & 0.33                     & 2      & 1    & verified  & 3.18      \\
\texttt{fair-termination}     & \texttt{koskinen-3-1}         & valid   & 0.17                     & 1      & 1    & verified  & 2.77      \\
\texttt{fair-termination}     & \texttt{koskinen-3-3}         & valid   & 0.18                     & 1      & 1    & verified  & 5.27      \\
\texttt{fair-termination}     & \texttt{koskinen-4}           & valid   & 0.24                     & 1      & 1    & verified  & 156.58    \\
\texttt{fair-termination}     & \texttt{lester}               & valid   & 0.31                     & 1      & 2    & timeout   & -         \\ \hline
\end{tabular}
\end{center}
\end{table}

\begin{table}[tbp]
\caption{Results of the experiment for comparison with previous higher-order program verification tools (2/2).}
\label{tab:resultextendedtwo}
\begin{center}
\begin{tabular}{|l|l|l|r|l|l|l|r|} \hline
                     &                      & \multicolumn{4}{l|}{\muhfl{}} & \multicolumn{2}{l|}{previous   tools} \\ \hline
benchmark            & instance             & result  & \multicolumn{1}{l|}{time} &
  \begin{tabular}[c]{@{}l@{}}no. of iter.\\ (prover)\end{tabular} & \begin{tabular}[c]{@{}l@{}}no. of iter.\\ (disprover)\end{tabular} & result & \multicolumn{1}{l|}{time} \\ \hline
\texttt{fair-non-termination} & \texttt{call-twice}           & valid   & 0.28                     & 1      & 1    & verified  & 1.05      \\
\texttt{fair-non-termination} & \texttt{compose}              & valid   & 0.13                     & 1      & 1    & verified  & 0.86      \\
\texttt{fair-non-termination} & \texttt{intro}                & valid   & 0.25                     & 1      & 2    & verified  & 3.99      \\
\texttt{fair-non-termination} & \texttt{loop-CPS}             & valid   & 0.12                     & 1      & 1    & verified  & 1.55      \\
\texttt{fair-non-termination} & \texttt{loop}                 & valid   & 0.12                     & 1      & 1    & verified  & 0.95      \\
\texttt{fair-non-termination} & \texttt{murase-closure-buggy} & valid   & 0.19                     & 1      & 2    & verified  & 1.01      \\
\texttt{fair-non-termination} & \texttt{murase-repeat-buggy}  & valid   & 0.20                     & 1      & 1    & verified  & 1.20      \\
\texttt{fair-non-termination} & \texttt{nested-if}            & valid   & 0.35                     & 1      & 3    & verified  & 1.47      \\
\texttt{fair-non-termination} & \texttt{odd-nonterm}          & valid   & 0.23                     & 1      & 2    & timeout   & -         \\
\texttt{fair-non-termination} & \texttt{op-loop}              & valid   & 0.19                     & 1      & 2    & verified  & 1.63      \\
\texttt{fair-non-termination} & \texttt{update-max}           & valid   & 0.43                     & 1      & 3    & verified  & 1.11      \\
\texttt{fair-non-termination} & \texttt{update-max-CPS}       & valid   & 0.35                     & 1      & 3    & verified  & 1.76      \\
\texttt{fair-non-termination} & \texttt{koskinen-1-buggy}     & valid   & 0.14                     & 1      & 1    & verified  & 3.55      \\
\texttt{fair-non-termination} & \texttt{koskinen-2-buggy}     & valid   & 0.16                     & 1      & 1    & verified  & 4.40      \\
\texttt{fair-non-termination} & \texttt{koskinen-3-1-buggy}   & valid   & 0.09                     & 1      & 1    & verified  & 2.07      \\
\texttt{fair-non-termination} & \texttt{koskinen-3-3-buggy}   & valid   & 0.21                     & 1      & 2    & verified  & 9.34      \\
\texttt{termination-ho}       & \texttt{ackermann-ho}         & timeout & -                        & 4      & 3    & fail      & -         \\
\texttt{termination-ho}       & \texttt{any-down-ho}          & valid   & 0.36                     & 1      & 1    & timeout   & -         \\
\texttt{termination-ho}       & \texttt{append-ho}            & valid   & 0.54                     & 1      & 1    & timeout   & -         \\
\texttt{termination-ho}       & \texttt{binomial-ho}          & valid   & 0.95                     & 1      & 1    & fail      & -         \\
\texttt{termination-ho}       & \texttt{fibonacci-ho}         & valid   & 1.04                     & 1      & 1    & timeout   & -         \\
\texttt{termination-ho}       & \texttt{foldr-ho}             & timeout & -                        & 4      & 8    & timeout   & -         \\
\texttt{termination-ho}       & \texttt{indirect-ho}          & valid   & 0.59                     & 1      & 1    & timeout   & -         \\
\texttt{termination-ho}       & \texttt{indirectHO-ho}        & valid   & 0.29                     & 1      & 1    & timeout   & -         \\
\texttt{termination-ho}       & \texttt{indirectIntro-ho}     & timeout & -                        & 5      & 3    & timeout   & -         \\
\texttt{termination-ho}       & \texttt{loop2-ho}             & timeout & -                        & 5      & 3    & timeout   & -         \\
\texttt{termination-ho}       & \texttt{map-ho}               & timeout & -                        & 4      & 9    & timeout   & -         \\
\texttt{termination-ho}       & \texttt{mc91-ho}              & timeout & -                        & 4      & 3    & timeout   & -         \\
\texttt{termination-ho}       & \texttt{mult-ho}              & valid   & 0.67                     & 1      & 1    & timeout   & -         \\
\texttt{termination-ho}       & \texttt{nested-loop-ho}       & valid   & 0.64                     & 1      & 1    & timeout   & -         \\
\texttt{termination-ho}       & \texttt{partial-ho}           & valid   & 0.32                     & 1      & 1    & timeout   & -         \\
\texttt{termination-ho}       & \texttt{quicksort-ho}         & timeout & -                        & 3      & 3    & timeout   & -         \\
\texttt{termination-ho}       & \texttt{sum-ho}               & valid   & 1.05                     & 1      & 1    & timeout   & -         \\
\texttt{termination-ho}       & \texttt{toChurch-ho}          & timeout & -                        & 5      & 9    & timeout   & -         \\
\texttt{termination-ho}       & \texttt{up-down-ho}           & timeout & -                        & 5      & 3    & timeout   & -         \\
\texttt{termination-ho}       & \texttt{x-plus-2-n-ho}        & timeout & -                        & 5      & 3    & timeout   & -         \\
\texttt{termination-ho}       & \texttt{zip-ho}               & valid   & 0.60                     & 1      & 1    & timeout   & -         \\
\texttt{fair-termination-ho}  & \texttt{murase-closure-ho}    & valid   & 0.68                     & 1      & 1    & timeout   & -         \\
\texttt{fair-termination-ho}  & \texttt{murase-intro-ho}      & valid   & 0.43                     & 1      & 1    & timeout   & -         \\
\texttt{fair-termination-ho}  & \texttt{murase-repeat-ho}     & valid   & 2.02                     & 1      & 1    & timeout   & -         \\
\texttt{fair-termination-ho}  & \texttt{hofmann-2-ho}         & valid   & 0.15                     & 1      & 1    & verified  & 3.25      \\
\texttt{fair-termination-ho}  & \texttt{koskinen-1-ho}        & timeout & -                        & 4      & 8    & fail      & -         \\
\texttt{fair-termination-ho}  & \texttt{koskinen-2-ho}        & timeout & -                        & 3      & 3    & timeout   & -         \\
\texttt{fair-termination-ho}  & \texttt{koskinen-3-1-ho}      & timeout & -                        & 6      & 3    & timeout   & -         \\
\texttt{fair-termination-ho}  & \texttt{koskinen-3-3-ho}      & valid   & 0.64                     & 1      & 1    & timeout   & -         \\
\texttt{fair-termination-ho}  & \texttt{koskinen-4-ho}        & valid   & 1.10                     & 1      & 1    & verified  & 6.33      \\
\texttt{fair-termination-ho}  & \texttt{lester-ho}            & valid   & 0.70                     & 1      & 1    & timeout   & -         \\  \hline
\end{tabular}
\end{center}
\end{table}
}

\end{document}